\documentclass[a4paper,UKenglish,cleveref, autoref, thm-restate]{lipics-v2021}

\newcommand{\tw}{\mathsf{tw}}
\newcommand{\br}[1]{\left(#1\right)}

\newcommand{\eps}{\varepsilon}

\newcommand{\Oh}{\mathcal{O}}
\newcommand{\poly}{\ensuremath{\mathsf{poly}}}

\newcommand{\diam}{\ensuremath{\mathsf{diam}}}

\newcommand{\inner}{\ensuremath{\mathsf{inner}{\text -}\mathsf{dist}}}
\newcommand{\loc}{\ensuremath{\mathsf{loc}}}
\newcommand{\sub}{\subseteq}
\newcommand{\sm}{\setminus}

\newcommand{\rr}{\ensuremath{\mathbb{R}}}
\newcommand{\nn}{\ensuremath{\mathbb{N}}}

\newcommand{\fcal}{\ensuremath{\mathcal{F}}}

\newcommand{\scal}{\ensuremath{\mathcal{S}}}

\newcommand{\parent}{\ensuremath{\mathsf{parent}}}
\newcommand{\ecc}{\ensuremath{\mathsf{ecc}}}
\newcommand{\ccal}{\mathcal{C}}
\newcommand{\layer}{\mathsf{layer}}
\newcommand{\apxdist}{\ensuremath{\mathsf{apx}{\text -}\mathsf{dist}}}

\def\DEBUG{true}

\ifdefined\DEBUG{}

\def\rem#1{{\marginpar{\raggedright\scriptsize #1}}}
\newcommand{\micr}[1]{\rem{\textcolor{blue}{\(\bullet \) #1}}}

\newcommand{\verr}[1]{\rem{\textcolor{purple}{\(\bullet \) #1}}}

\newcommand{\karr}[1]{\rem{\textcolor{orange}{\(\bullet \) #1}}}
\else

\newcommand{\micr}[1]{ }
\newcommand{\verr}[1]{ }
\newcommand{\karr}[1]{ }
\fi

\title{Going Beyond Surfaces in Diameter Approximation}

\author{Michał Włodarczyk}{University of Warsaw, Poland}{michal.wloda@gmail.com}{https://orcid.org/0000-0003-0968-8414}{}

\authorrunning{M. Włodarczyk} 

\Copyright{Michał Włodarczyk} 



\category{} 

\relatedversion{} 

\funding{Supported by Polish National Science Centre SONATA-19 grant 2023/51/D/ST6/00155.}

\hideLIPIcs
\nolinenumbers 
\EventEditors{Anne Benoit, Haim Kaplan, Sebastian Wild, and Grzegorz Herman}
\EventNoEds{4}
\EventLongTitle{33rd Annual European Symposium on Algorithms (ESA 2025)}
\EventShortTitle{ESA 2025}
\EventAcronym{ESA}
\EventYear{2025}
\EventDate{September 15--17, 2025}
\EventLocation{Warsaw, Poland}
\EventLogo{}
\SeriesVolume{351}
\ArticleNo{37}

\begin{document}

\maketitle

\begin{abstract}
Calculating the diameter of an undirected graph requires quadratic running time under the Strong Exponential Time Hypothesis and this barrier works even against any approximation better than~$3/2$.
For planar graphs with positive edge weights, there are known $(1+\eps)$-approximation algorithms with running time $\poly(1/\eps,\, \log n)\cdot n$.
However, these algorithms rely on shortest path separators 
and this technique falls short to yield efficient algorithms beyond graphs of bounded genus.

In this work we depart from embedding-based arguments and obtain diameter approximations relying on VC set systems and the local treewidth property.
We present two orthogonal extensions of the planar case by giving $(1+\eps)$-approximation algorithms with the following running times:
\begin{itemize}
    \item $\Oh_h((1/\eps)^{\Oh(h)}\cdot n\log^2 n)$-time algorithm for graphs excluding an apex graph of size $h$ as a~minor,
    \item $\Oh_d((1/\eps)^{\Oh(d)}\cdot n\log^2 n)$-time algorithm for the class of $d$-apex graphs.
\end{itemize}

As a stepping stone, we obtain efficient $(1+\eps)$-approximate distance oracles for graphs excluding an apex graph of size $h$ as a minor.
Our oracle has preprocessing time $\Oh_h((1/\eps)^8\cdot n\log n\log W)$ and query time $\Oh_h((1/\eps)^2\cdot \log n\log W)$, where $W$ is the metric stretch.
Such oracles have been so far only known for bounded genus graphs.
All our algorithms are deterministic.
\end{abstract}

\section{Introduction}

The diameter of a graph $G$, denoted $\diam(G)$, is the largest distance between a pair of its vertices in the shortest path metric.
We assume that each edge is assigned some real positive weight, and the length of a path is the sum of its weights.
The obvious way to calculate $\diam(G)$ is to check the distances for all vertex pairs, but this inevitably requires quadratic time. 
This trivial approach is essentially optimal because already unweighted graphs of diameter 2 or 3 cannot be distinguished in time $\Oh(m^{1.99})$ under the Strong Exponential Time Hypothesis~\cite{Roditty13W} (see~\cite{AbboudWW16} for hardness in sparse graphs).
On the other hand, truly subquadratic algorithms to compute diameter are available in planar graphs~\cite{Cabello19, GawrychowskiKMS21} and graph classes excluding a minor~\cite{ducoffe2022diameter, kluk2025faster, LeW24}. 
Still, the best algorithm for 
planar graphs runs in time\footnote{
The variable $n$ refers to the number of vertices in a graph but all the considered graph classes are sparse so the number of edges is $\Oh(n)$.
The $\tilde\Oh$-notation hides factors polylogarithmic in $n$. We refer to the dependence $\tilde\Oh(n)$ as {\em almost-linear}.} $\tilde \Oh(n^{5/3})$~\cite{GawrychowskiKMS21}, which is far from linear.

The absence of fast algorithms calculating the diameter exactly led the researchers to seek fast approximation algorithms. 
Weimann and Yuster~\cite{weimann2015approximating} gave the first almost-linear $(1+\eps)$-approximation for planar diameter, running in time $\tilde\Oh(2^{\Oh(1/\eps)} \cdot n)$.
This was improved to randomized time $\tilde\Oh((1/\eps)^{5} \cdot n)$ by Chan and Skrepetos~\cite{ChanS17} (see also~\cite{ChangKT22}).
Li and Parter~\cite{li2019planar} gave a deterministic algorithm of this type, as well as a distributed algorithm with a small number of rounds.
Observe that each of these can be used to recognize unweighted planar graphs of diameter $D$, simply by setting $\eps = 1/(2D)$.
There are also specialized algorithms for this task by Eppstein~\cite{Eppstein00}, with running time $2^{\Oh(D \log D)}\cdot n$, and by Abboud, Mozes, and Weimann~\cite{abboud2023what}, with running time $\tilde\Oh(D^2\cdot n)$.

\subparagraph{Shortest path separators.}
An area that is technique-wise related is the design of approximate distance oracles.
After a preprocessing step, we would like to quickly estimate any $(u,v)$-distance within factor $(1+\eps)$.
Thorup~\cite{Thorup04} exploited the fact that planar graphs admit balanced separators consisting of two shortest paths to give an oracle with preprocessing time $\Oh((1/\eps)^2\cdot n\log^3 n)$, space $\Oh((1/\eps)\cdot n\log n)$, and query time  $\Oh(1/\eps)$.
The shortest path separators constitute a crucial ingredient of all the three aforementioned algorithms for estimating planar diameter as well.  

Abraham and Gavoille~\cite{abraham2006object} presented a generalization of shortest path separators to graph classes excluding a minor and extended Thorup's oracle to such classes.
However, their separators no longer enjoy such a nice structure.
The main caveat is that already in a single balanced separator one might need to use paths $P_1,P_2$ such that $P_2$ is a shortest path only in $G-P_1$.
Hence $P_2$ may be much longer than $\diam(G)$, which is easily seen to be unavoidable already in the class of apex graphs\footnote{$G$ is called an {\em apex graph} if $G$ can be made planar by removal of a single vertex.}, and it is unclear how to exploit such a separator for diameter computation.
The second issue is that finding the separator requires highly superlinear time, which burdens the preprocessing time of the oracle.
Almost-linear preprocessing time is currently available only for graphs of bounded genus~\cite{kawarabayashi2011linear} and for unit disk graphs~\cite{ChanSkrepetos2019}.
{The lack of well-structured balanced separators beyond surface-embeddable graphs seems to have limited the development of efficient metric-related algorithms so far.}

\subparagraph{Beyond surface embedding.}
There are two directions of generalizing the topological graph classes.
First, one can pick $\Oh(1)$ faces in the surface embedding of a graph and replace them with certain graphs of pathwidth $\Oh(1)$, which are called the {\em vortices}.
Second, one can insert a set of $\Oh(1)$ vertices, called the {\em apices}, and connect them possibly to anyone.
The structural theorem of Robertson and Seymour~\cite{robertson2003graph} states that for any graph $H$, all the graphs excluding $H$ as a minor can be constructed from graphs embeddable in some fixed surface using these two operations and by taking a clique-sum of them.

When the excluded minor $H$ is an apex graph, this description can be slightly restricted~\cite{DemaineHK09}. 
What is important to us, apex-minor-free graphs enjoy the {\em local treewidth property}:
their treewidth is bounded in terms of the (unweighted) diameter and this relation is known to be linear~\cite{DemaineH04}.
Note that the diameter of an edge-weighted graph cannot be related to treewidth by any means because the diameter can be reduced by scaling the weights, while the topological structure of the graph remains the same.
Nevertheless, we will see that the local treewidth property is still useful in the considered setting.

{\bf The main message of this work is that the local treewidth property can be used instead of shortest path separators for the sake of approximating diameter and radius, as well as for efficient construction of distance oracles.}
We first employ it to tackle apex-minor-free graphs but we also show that its usefulness is not limited to such graph classes.

\subparagraph{Additive core-sets.}
Suppose we have a list of vertices $S = (s_0,s_1,\dots,s_{k-1})$ in $G$ and we want to (approximately) learn all the distance patterns, between a vertex $v$ and all of $S$. 
Chan and Skrepetos~\cite{ChanS17} employed Cabellos's technique of abstract Voronoi diagrams~\cite{Cabello19} to build a data structure to browse such distance patterns in the case when $G$ is planar and $S$ lie on a single face.
However, this technique relies heavily on the embedding of the graph.

A more versatile argument was given by Li and Parter~\cite{li2019planar} in order to approximate planar diameter in a distributed regime.
Following their convention, 
a subset $U' \sub U$ is called a {\em $\delta$-additive core-set} of $U$
if for every $u \in U$ there exists $u' \in U'$ such that for each $i \in [0,k)$ we have $|d(u,s_i) - d(u',s_i)| \le \delta$.
When $\delta = \eps\cdot \diam(G)$, a~trivial rounding argument yields a {$\delta$-additive core-set} of size $(1/\eps)^k$.
Remarkably, the dependence on $k$ can be in fact removed from the exponent~\cite{li2019planar}.
This theorem relies on a bound on the VC dimension of a certain set system and this argument carries over to any minor-free graph class~\cite{LeW24}.
A~direct consequence of these theorems reads as follows.
(Here $K_h$ stands for the $h$-vertex clique.)

\begin{restatable}[{Corollary of {\cite[Theorem 1.3]{LeW24}}}]{lemma}{lemCoreSet}
\label{lem:corest-exist}
    Let $G$ be an edge-weighted $K_h$-minor-free graph~$G$, $S = (s_0,\dots,s_{k-1}) \in V(G)^k$, $U \sub V(G)$, and $\eps > 0$.
    Then $U$ admits an $(\eps \cdot \diam(G))$-additive core-set (with respect to $S$) of size $\Oh(k^{h-1}\cdot (1/\eps)^{h})$.
\end{restatable}

For the sake of completeness, we provide a proof in \Cref{sec:prelim}.
The local treewidth property will allow us to apply this lemma for $k = \poly(1/\eps)$  to efficiently compress information about distances in large subgraphs.

\subsection{Our contribution}

A crucial property of a separator $S$ consisting of $\Oh(1)$ shortest paths is that $S$ can be covered by $\Oh(1/\eps)$ subgraphs of strong\footnote{A subgraph $H$ of $G$ has {\em strong diameter} $D$ if every pair of vertices $u,v \in V(H)$ can be connected by a~path in $H$ of length at most $D$. For comparison, {\em weak diameter} refers to distances measured in $G$.} diameter $\eps\cdot\diam(G)$.
Then the interaction between the two sides of $S$ can be channeled through $\Oh(1/\eps)$ {\em portals} on $S$ when small distance distortion is allowed.
We show that balanced separators of this kind exist and can be efficiently found in any class excluding an apex minor.
Specifically, we construct a clustering of an apex-minor-free graph into pieces of small strong diameter such that the graph obtained by contracting each cluster has bounded treewidth.
We use this construction to estimate the eccentricity $\ecc_G(v)$ of every vertex $v$, that is, $\max_{u \in V(G)} d_G(v,u)$.

\begin{restatable}{theorem}{thmFirst}
\label{thm:first}
    Let $H$ be an apex graph of size $h$. There is  an algorithm that, given an edge-weighted $H$-minor-free graph $G$ and $\eps >0$, runs in time $\Oh_{h}((1/\eps)^{3h+2} \cdot n\log^2 n)$ and estimates eccentricity of each vertex vertex $v \in V(G)$ within additive error $\eps\cdot\diam(G)$.
    Moreover, for each $v \in V(G)$ the algorithm outputs vertex $v'$ for which $d_G(v,v') \ge \ecc_G(v) - \eps\cdot\diam(G)$. 
\end{restatable}

As a corollary, we get $(1+\eps)$-approximation for both diameter and radius, 
which are
 respectively the largest and the smallest eccentricity. 

To employ the additive core-sets in the proof of \Cref{thm:first}, we first need to precompute approximate distances for $\Oh((1/\eps)^2\cdot n\log n)$ pairs of vertices.
To this end, we design a distance oracle with preprocessing and query times tailored to our needs.
Similar oracles are available for minor-free graphs~\cite{abraham2006object, chang2024shortcut, kawarabayashi2011linear} but their preprocessing times are far from linear.
An~oracle with additive error $\eps\cdot\diam(G)$ suffices for
the proof of \Cref{thm:first} but we can improve it to multiplicative error by taking into account the metric stretch, i.e., the ratio between the largest and smallest positive distances. 

\begin{restatable}{theorem}{thmOracle}
\label{thm:oracle}
    Let $H$ be an apex graph. 
    Given an edge-weighted $H$-minor-free graph $G$ of metric stretch $W$ and $\eps >0$, one can build an $(1+\eps)$-approximate distance oracle with
    preprocessing time $\Oh_{H}((1/\eps)^{8}\cdot  n\log n\log W)$ 
    and query time $\Oh_H((1/\eps)^{2}\cdot\log n\log W)$.
\end{restatable}

The simplest case not covered by \Cref{thm:first} is of course given by the class of apex graphs.
This class no longer enjoys the local treewidth property.
Moreover, some problems which are easy on planar graphs become much harder after insertion of a single apex~\cite{barahona1983max, ChakrabortyG23}.
Note that removing just a single vertex can  increase the diameter significantly and so our clustering technique cannot be directly applied to decompose such a graph.


We give a stronger theorem that in particular covers the apex graphs. 
As a consequence of Robertson and Seymour's structural theorem, for every graph $K$ there exists an apex graph $H$ and a constant $c$, such that every $K$-minor-free graph can be constructed via a clique-sum from graphs $G_1,\dots,G_\ell$
with the following property: each $G_i$ becomes $H$-minor-free after deletion of $c$ vertices~{\cite[Theorem VI.1]{CohenAddadLPP23}}.
We extend \Cref{thm:first} to such graphs, assuming that the deletion set is known.

\begin{restatable}{theorem}{thmApexDiam}
\label{thm:apex}
    Let $h,c \in \nn$ and $\cal H$ be a graph class excluding an apex graph of size $h$. There is an algorithm that, given an edge-weighted graph $G$, $\eps >0$, and a set $A \sub V(G)$ of size $c$, for which $G-A \in \cal H$, runs in time $\Oh_{h+c}((1/\eps)^{\alpha}\cdot n\log^2 n)$, where $\alpha = \Oh((h+1)(c+1))$, and estimates $\diam(G)$ within factor
     $(1+\eps)$.
\end{restatable}

Unlike the apex-minor-free case, this time we do not estimate eccentricities  because \Cref{thm:apex} 
involves a certain ``irrelevant vertex rule'' that 
only preserves the diameter (approximately). 
The assumption that the set $A$ is given in the input can be dropped for the class $\cal H$ of planar graphs, i.e., when the input is a $d$-apex graph.
In this case, one can find the set $A$ in time $2^{\Oh(d \log d)}\cdot n$~\cite{JansenLS14}, which leads to the following unconditional statement.

\begin{corollary}
\label{cor:planar}
        There is an $\Oh_{d}((1/\eps)^{\Oh(d)} \cdot n\log^2 n)$-time algorithm that estimates the diameter of an~edge-weighted $d$-apex graph within factor $(1+\eps)$. 
\end{corollary}

This result can be regarded as a generalization of the planar case in a direction that is orthogonal to \Cref{thm:first}.
Unfortunately, when the class $\cal H$ excludes an arbitrary apex graph $H$,
we do not know whether one can drop the assumption about the set $A$ being provided.
Indeed, the fastest parameterized algorithm to find set $A$ for which $G-A$ is $H$-minor-free (where $H$ is some apex graph) requires quadratic time~\cite{SauST22}.
We discuss the perspective of handling more general graph classes in \Cref{sec:conclusion}. 

\subsection{Related Work}

Apart from the ones used in this work, there are more VC set systems considered in the context of minor-free graph metrics
~\cite{bousquet2015vc, braverman2021coresets, ChepoiEV07, Cohen-AddadD0SS25, ducoffe2022diameter, KarczmarzZ25}.
Some other approaches to such metrics
are based on a decomposition into low-diameter subgraphs, including the classic KPR theorem~\cite{KPR93}.
Our low-diameter clustering scheme 
bears resemblance to 
the {\em shortcut partitions}~\cite{chang2023covering, chang2023resolving, chang2024shortcut} which provide even stronger guarantees but it is not known how to construct them in almost-linear time.
The related {\em buffered cop decompositions}~\cite{chang2024shortcut} yield cluster graphs of bounded weak coloring numbers~\cite{BourneufP25}.
These in turn are connected to {\em sparse covers}~\cite{abraham2007strong, busch2007improved, Filtser2004sparse}, however there the low-diameter subgraphs need not to be disjoint.

An alternative way to represent a metric by a bounded treewidth graph is a~{\em metric embedding}.
For example, an edge-weighted planar graph $G$ can be embedded into a metric given by a graph of treewidth $f(\eps)$ so that the additive distortion is $\eps \cdot \diam(G)$~\cite{FoxEpsteinKS19}.
The function $f$ can be estimated as $f(\eps) = \Oh(\eps^{-4})$~\cite{chang2023covering}.
For apex-minor-free graphs, the current best treewidth bound is $f(\eps) = 2^{\Oh(1/\eps)}$~\cite{chang2024shortcut}.
In the latter case, we again do not know if such an embedding can be computed in almost-linear time.
See also \cite{CohenAddadLPP23, Arnold22} for metric embeddings of minor-free graphs and a variant with multiplicative distortion. 

For more  advances in diameter computation, see \cite{ChanSkrepetos2019, chang24, duraj2023better}, and for recent developments in planar distance oracles, see \cite{ChangKT22, charalampopoulos2023almost, fredslundhansen, gawrychowski2018better, Le23}.

\subsection{Organization of the Paper}

We begin with \Cref{sec:prelim} including preliminaries on VC set systems and additive core-sets.
The following three sections contain the main technical lemmas.
Sections \ref{sec:apex-free} and \ref{sec:tree} culminate with the proof of Theorem \ref{thm:first}.
\Cref{sec:apices} contains the crucial lemma for \Cref{thm:apex}, which is then proven in \Cref{app:main}, alongside \Cref{thm:oracle}. 
The proofs of the statements marked with~$(\bigstar)$ are located in \Cref{app:missing}.

\section{Preliminaries}
\label{sec:prelim}

All our statements concern undirected graphs with real positive edge weights.
In an unweighted graph, each weight is set to one. 

\begin{definition}[Treewidth]\label{def:prelim:tw}
A \emph{tree decomposition} of a graph $G$ is a pair $(T, \chi)$ where $T$ is a tree, and~$\chi \colon V(T) \to 2^{V(G)}$ is a function, such that:
\begin{enumerate}
    \item for each~$v \in V(G)$ the nodes~$\{t \mid v \in \chi(t)\}$ form a {non-empty} connected subtree of~$T$, \label{item:tree:based:connected}
    \item for each edge~$uv \in E(G)$ there is a node~$t \in V(T)$ with~$\{u,v\} \subseteq \chi(t)$. \label{item:tree:based:edge}
\end{enumerate}
{The \emph{width} of $(T, \chi)$ is defined as~$\max_{t \in V(T)} |\chi(t)| - 1$.}
The \emph{treewidth} of a graph $G$ (denoted $\tw(G))$ is the minimal width a tree decomposition of $G$.
\end{definition}


\subparagraph*{VC dimension.}
A set system is a collection $\fcal$ of subsets of some set $U$.
We say that a subset $Y \sub A$ is {\em shattered} by $\fcal$ if $\{Y \cap S \colon S \in \fcal\} = 2^Y$ , that is,
every subset of $Y$ is an intersection of $Y$ with some $ S \in \fcal$. The VC dimension of a set system is the size of the largest  shattered subset. The notion of VC dimension
was introduced by Vapnik and Chervonenkis~\cite{vapnik2015uniform}. What is important to us, a set system of bounded VC dimension cannot be too large.

\begin{lemma}[Sauer–Shelah Lemma \cite{sauer1972density}]
\label{lem:size-bound}
    Let $\fcal$ be a set system  of VC dimension $d$ over a set of size $m$.
    Then $|\fcal| = \Oh(m^d)$.
\end{lemma}

\subsection{Additive core-sets}

For a vector $\bf x \in \mathbb{R}^k$ its $\ell_\infty$-norm is defined as $\ell_\infty({\bf x}) = \max_{i=1}^k |x_i|$.

\begin{definition}[\cite{li2019planar}]
For a sequence of vertices $S = (s_0,\dots,s_{k-1})$ from $G$ we define the {\em distance tuple} of $v$ with respect to $S$ as $\Phi_S(v) = \left(d_G(v,s_0), \dots, d_G(v,s_{k-1})\right)$.

For a vertex set $V$, a subset $V' \sub V$ is called a {\em $\delta$-additive core-set} with respect to $S$ if for every $v \in V$ there exists $v' \in V'$ such that $\ell_\infty(\Phi_S(v) - \Phi_S(v')) \le \delta$.
\end{definition}

We will take advantage of the bound on VC dimension of certain set systems related to a~graph metric.
Le and Wulff-Nilsen~\cite{LeW24} considered a slight modification of the set systems introduced by Li and Parter~\cite{li2019planar}, and so they denoted them by ${\cal \widehat{LP}}$.

\begin{definition}[\cite{LeW24}]
\label{def:vc}
    Let $G$ be an  edge-weighted graph, $S$ be a $k$-tuple of vertices from $G$, and $M \sub \rr$.
    For $v \in V(G)$ we define:

\[
\hat{X}_v = \{(i,\Delta) \colon i \in [1,k-1],\, \Delta \in M,\, d_G(v,s_i) - d_G(v,s_0) \le \Delta\}. 
\]
Let ${\cal \widehat{LP}}_{G,M}(S) = \{{\hat X}_v \colon v \in V(G)\}$ be a collection of subsets of the ground set $[k - 1] \times M$.
\end{definition}

\begin{theorem}[{\cite[Theorem 1.3]{LeW24}}]
\label{lem:vc-bound}
    Let $G$ be an  edge-weighted $K_h$-minor free graph and $S,M$ be as in \Cref{def:vc}.
    Then ${\cal \widehat{LP}}_{G,M}(S)$ has VC dimension at most $h-1$.
\end{theorem}

We adapt the construction of additive core-sets from~\cite{li2019planar} to rely on the set system ${\cal \widehat{LP}}$.

\lemCoreSet*
\begin{proof}
    Let $D \ =\diam(G)$, $\delta = \eps\cdot D/2$, and $z = \lceil 2/\eps \rceil$.
    Next, we choose $M = \{j\cdot \delta \colon j \in [-z,z]\}$.
    We will extend \Cref{def:vc} slightly.
    For vertices $u,v$ we define $d'(u,v)$ as the largest integer $j$ for which $j\cdot\delta \le d_G(u,v)$.
    Note that $d'(u,v)$ is always bounded by $z$.
    For a vertex $v$ let ${\hat X}^\circ_v$ be the pair $(d'(v,s_0), {\hat X}_v)$.
    By \Cref{lem:vc-bound} and \Cref{lem:size-bound}, the family $\{{\hat X}_v \colon v \in V(G)\}$ has size $\Oh\br{((k-1)\cdot |M|)^{h-1}} = \Oh\br{(k/ \eps)^{h-1}}$.
    It follows that the family $\{{\hat X}^\circ_v \colon v \in V(G)\}$ has  $\Oh\br{k^{h-1} \eps^{-h}}$ elements. 

    We construct $U'$ by picking for each $X \in \{{\hat X}^\circ_v \colon v \in V(G)\}$ an arbitrary vertex $v \in U$ with ${\hat X}^\circ_v = X$.
    Then $|U'| = \Oh\br{k^{h-1} \eps^{-h}}$.
    To show that $U'$ is an $(\eps D)$-additive core-set, pick $u \in U$.
    We need to prove that there exists $u' \in U'$ such that $\ell_\infty(\Phi_S(u) - \Phi_S(u')) \le \eps D$.
    By construction, there is $u' \in U'$ for which ${\hat X}^\circ_u = {\hat X}^\circ_{u'}$.
    Now we need to argue that for each $i \in [0,k-1]$ it holds that $|d_G(u,s_i) - d_G(u',s_i)| \le \eps D$. 

    The simplest case is for $i = 0$.
    As $d'(u,s_0) = d'(u',s_0)$, there is an integer $j$ such that both $d_G(u,s_0),\, d_G(u',s_0)$ belong to the interval $[j\cdot\delta,\, (j+1)\cdot\delta)$.
    Hence $|d_G(u,s_0) - d_G(u',s_0)| \le \delta = \eps D/2$. 

    Now consider $i \ge 1$.
    Let $j$ be the largest integer (possibly negative) for which $d_G(u,s_i) - d_G(u,s_0) \le j \cdot \delta$.
    We define $j'$ analogously with respect to $u'$.
    Note that $j,j' \in [-z,z]$. 
    Because ${\hat X}_u = {\hat X}_{u'}$ it must be $j=j'$.
    By the same argument as before, the numbers $x = d_G(u,s_i) - d_G(u,s_0)$, $y = d_G(u',s_i) - d_G(u',s_0)$ differ by at most $\delta$.
    Finally, since $|d_G(u,s_0) - d_G(u',s_0)| \le \delta$ we obtain that $|d_G(u,s_i) - d_G(u',s_i)| \le 2\delta = \eps D$.
\end{proof}

In our setting we will only know approximations of the distances from $S$.
We show that this suffices to construct an additive core-set.

\begin{restatable}{lemma}{lemCoreApx}
\label{lem:metric:apx-cover}
    Let $G$ be an edge-weighted $K_h$-minor-free graph $G$, $S = (s_0,\dots,s_{k-1}) \in V(G)^k$, and $U \sub V(G)$.
    Let $\delta$ be given and denote $\eps = \delta / \diam(G)$.
    Next, suppose that for each $v \in U$ we are given ${\bf y}_v \in \rr^k$ that satisfies $\ell_\infty({\bf y}_v - \Phi_S(v)) \le \delta$. 
    Then in time $\Oh((k / \eps)^{h} \cdot |U|)$ we can compute a $(5\delta)$-additive core-set of $U$ of size $\Oh\br{k^{h-1} \eps^{-h}}$.
\end{restatable}
\begin{proof}
    We first describe the algorithm, which is a simple greedy construction of a~$3\delta$-cover in a metric space.
    Initialize $\hat U = \emptyset$ and
    iterate over $v \in U$ in any fixed order.
    When processing $v$, check if $\hat U$ contains $u$ for which $\ell_\infty({\bf y}_v - {\bf y}_u) \le 3\delta$.
    If there is no such element, add $v$ to $\hat U$.

    First we show that every vector in $\{{\bf y}_v \colon v\in U\}$ is within distance $3\delta$  in the $\ell_\infty$ metric from the set
    $\{{\bf y}_v \colon v\in \hat U\}$. 
    Suppose otherwise and let $v \in U$ be a vertex whose vector ${\bf y}_v$ is further than $3\delta$ from this set. 
    In particular, this means that $\ell_\infty({\bf y}_v - {\bf y}_u) > 3\delta$ for each $u$ considered before $v$.
    So $v$ must have been added to $\hat U$ at this point; a contradiction.

    Now we argue that $\hat U$ is a $(5\delta)$-additive core-set.
    For each $u \in U$ we know that there exists $v \in \hat U$ such that $\ell_\infty({\bf y}_v - {\bf y}_u) \le 3\delta$.
    Next, by assumption $\ell_\infty({\bf y}_u - \Phi_S(u)) \le \delta$ and likewise for $v$.
    Then by the triangle inequality we get $\ell_\infty(\Phi_S(v) - \Phi_S(u)) \le 5\delta$, as intended.

    Next, we bound the size $\hat U$.
    By construction, for each  $u, v \in \hat U$
    we have $\ell_\infty({\bf y}_v - {\bf y}_u) > 3\delta$.
    We apply the triangle inequality again to get $\ell_\infty(\Phi_S(v) - \Phi_S(u)) > \delta$.
    By applying \Cref{lem:corest-exist} for $\eps' = \eps/2$ we obtain $U' \sub U$ so that the $(\delta/2)$-radius balls (in the $\ell_\infty$ metric) centered at $\{\Phi_S(v) \colon v\in U'\}$ cover the set $\{\Phi_S(v) \colon v\in U\}$.
    Each such ball has diameter $\delta$ so it can contain at most one element from the set $\{\Phi_S(v) \colon v\in \hat U\}$.
    This implies $|\hat U| \le |U'|$ and so the size bound follows from \Cref{lem:corest-exist}.
    
    Finally, we analyze the running time.
    When inspecting a vertex $u \in U$ we need to compute distances in $\rr^k$ between ${\bf x}_u$ and the vectors of all the vertices already placed in the core-set. 
    This requires at most $k \cdot |\hat U|$ operations and so the running time bound follows from the bound on $|\hat U|$.
\end{proof}

\section{Low Radius Clustering}
\label{sec:apex-free}

We begin with summoning the local treewidth property of apex-minor-free graphs.

\begin{lemma}[\cite{DemaineH04}]
\label{lem:local-property}
    For every apex graph $H$, there exists a constant $c$ such that for every unweighted $H$-minor-free graph $G$ it holds that $\tw(G) \le c \cdot \diam(G)$.
\end{lemma}

Our first tool is a clustering of a graph  into induced subgraphs of small strong diameter.  

\begin{definition}\label{def:clustering}
    For an edge-weighted graph $G$, we define a {\em $\delta$-radius clustering} of $G$ as a collection of pairs $(C_i,c_i)_{i=1}^\ell$ with the following properties.
    \begin{enumerate}
        \item $(C_i)_{i=1}^\ell$ forms a partition of $V(G)$.
        \item For each $i \in [\ell]$ it holds that $G[C_i]$ is connected, $c_i \in C_i$, and each $v \in C_i$ is within distance~$\delta$ from $c_i$ in $G[C_i]$.\label{item:def:strong}
    \end{enumerate}
    For ${\cal C} = (C_i,c_i)_{i=1}^\ell$ we denote by $G_\ccal$ the {\em cluster graph} obtained from $G$ by contracting each $C_i$ into a single vertex.
    The vertex $c_i$ is called the {\em center} of the $i$-th cluster.
\end{definition}

Note that each cluster in a  $\delta$-radius clustering has strong diameter at most $2\delta$.
As our first goal, we aim to construct a $\delta$-radius  clustering, for $\delta = \eps\cdot\diam(G)$, whose cluster graph has treewidth $\Oh(1/ \eps)$.
Secondly, even when $\delta$ is chosen to be much smaller $\diam(G)$ we would like to put vertices $u,v$, satisfying $d_G(u,v) = \Oh(\delta)$, into ``nearby'' clusters.

For a graph $G$, a {\em layering} is simply a function $\layer \colon V(G) \to \nn$.
When such a function is clear from the context, for $I \sub \nn$ we denote by $G^I$ the subgraph of $G$ induced by the vertices $v$ with $\layer(v) \in I$.
Given a clustering $\ccal$ of $G$ and a 
layering of $G_\ccal$, we naturally extend it to a layering of~$G$. 

\begin{figure}
\centering\includegraphics[width=13cm]{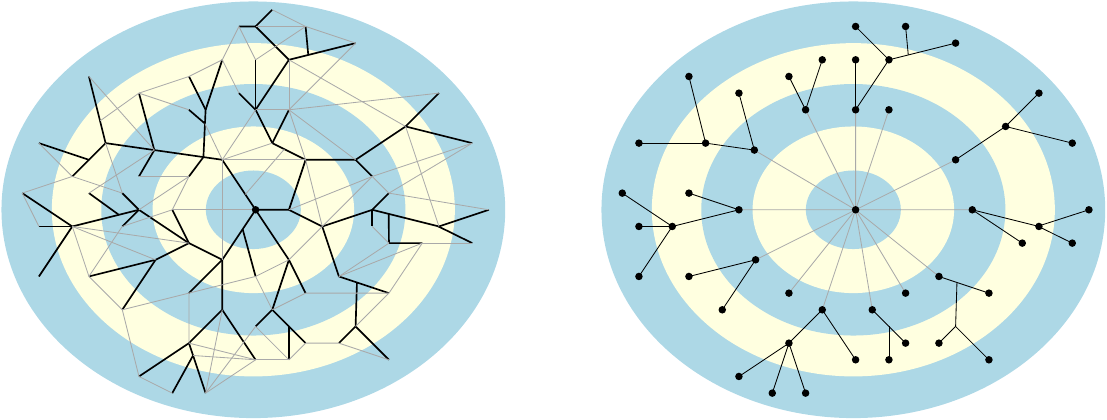}
\caption{Construction of the clustering in \Cref{lem:local-tw}.
Left: The black edges form a shortest-path tree $T$ rooted at $r$.
The alternating background colors represented the layers.
Each connected subgraph of $T$ contained within a single layer forms a  cluster.
Right: An argument that $G_\ccal^I$ has bounded treewidth for $I = \{2,3,4\}$.
For simplicity, the edges not originating from $T$ are omitted. 
After contracting the layers 0 and 1 into a single vertex, we obtain a graph of (unweighted) diameter at most 6.
Since this graph is apex-minor-free, we can use the local treewidth property.
}
\label{fig:cluster}
\end{figure}

\begin{restatable}[$\bigstar$]{lemma}{lemLocal}
\label{lem:local-tw}
    Let $H$ be an apex graph.
    Given a connected edge-weighted $n$-vertex $H$-minor-free graph $G$ and $\delta > 0$, we can in time $\Oh(n\log n)$ find a $\delta$-radius clustering $\ccal$ of $G$ and a layering function $\layer \colon V(G_\ccal) \to \nn$ such that the following hold.
    \begin{enumerate}
        \item For an interval $I \sub \nn$
        of size $p$, the graph $G^I_\ccal$ has treewidth $\Oh_H(p)$.\label{item:local:interval}
        \item The treewidth of $G_\ccal$ is bounded by $\Oh_H(\diam(G) / \delta)$.\label{item:local:tw}
        \item For each $u,v \in V(G)$ and $p \in \nn$ satisfying $d_G(u,v) \le \delta\cdot p$, it holds that $|\layer(u) - \layer(v)| \le p$. In this case, there is an interval $I \sub \nn$ depending only on $(\layer(u),\layer(v),p)$, such that $|I| \le 2p+1$ and the
        $(u,v)$-distance is the same in $G$ and~$G^I$. \label{item:local:distance}
    \end{enumerate}
\end{restatable}
\hspace{0.1cm}

The proof constructs the layering by analyzing a shortest-path tree $T$ rooted at an arbitrary vertex $r$.
Then the $j$-th layer corresponds to vertices whose distance from $r$ belongs to $[j\cdot\delta, \,j\cdot\delta + \delta)$.
The clustering $\ccal$ is obtained by contracting all the edges in $T$ connecting vertices sharing the same layer.

Clearly, each root-to-leaf path in $T$ traverses at most $\diam(G) / \delta$ layers, which bounds the (unweighted) diameter of $G_\ccal$.
Since $G_\ccal$ is a contraction of $G$, it is $H$-minor-free as well and we can take advantage of the local treewidth property.
A similar argument yields \Cref{item:local:interval} while \Cref{item:local:distance} follows from the triangle inequality.

\Cref{lem:local-property} only guarantees that $\tw(G_\ccal)$ can be bounded in terms of $\eps$, without providing an efficient way to construct a tree decomposition. 
Moreover, the known algorithms for computing optimal (or $\Oh(1)$-approximate) tree decomposition require time exponential in treewidth, whereas we aim at polynomial dependence on $\eps$.
Luckily, there is one algorithm available that suits our needs, with a bit worse width guarantee.

While several algorithms rely on dynamic programming over a tree decomposition to compute the diameter exactly~\cite{AbboudWW16, Eppstein00, Husfeldt16}, we cannot afford such an strategy due to prohibitive error accumulation.
Instead, we will turn the decomposition into a balanced one, 
again by increasing its width slightly, and use it to build our data structures. 

\begin{lemma}\label{lem:prelim:twd-compute}
    There is an algorithm that, given a graph $G$ of treewidth at most $k$, runs in time $\Oh(k^7\cdot n\log n)$ and outputs a binary rooted tree decomposition of $G$ of width $\Oh(k^2)$, depth $\Oh(\log n)$, and size $\Oh(n)$.
\end{lemma}
\begin{proof}
    First, in time $\Oh(k^7\cdot n\log n)$ one can compute a tree decomposition $(T,\chi)$ of $G$ of width $k' = \Oh(k^2)$~\cite{FominLSPW18}.
    By a standard argument, we can modify $(T,\chi)$ to have size $\Oh(n)$.
    Next, we turn $(T,\chi)$ into a binary rooted tree decomposition of depth $\Oh(\log n)$ and width at most $4k' + 3$ in time $\Oh(n)$~\cite{chatterjee2014optimal}.
    This also preserves the bound $|V(T)| = \Oh(n)$.
\end{proof}

\section{Processing a Tree Decomposition}
\label{sec:tree}

We will need an  argument to reroute paths to go through a center $c_i$ of some cluster $C_i$.

\begin{lemma}\label{lem:path-touch}
    Let $G$ be an edge-weighted graph, $u,v,w \in V(G)$, and let $C$ be a vertex set contained in the $\delta$-radius ball around $w$.
    Suppose there is a shortest $(u,v)$-path $P$ in $G$ that intersects $C$.
    Then $d_G(u,w) + d_G(w,v) \le d_G(u,v) + 2\delta$.
\end{lemma}
\begin{proof}
    Orient the path $P$ from $u$ to $v$ and define $u',v'$ as the first and the last vertices from $C$ appearing on $P$.
    Since $u',v'$ lie on a shortest $(u,v)$-path, we have $d_G(u,u') + d_G(v',v) \le d_G(u,v)$.
    Next, it holds that each of $d_G(w,u')$, $d_G(w,v')$ is at most $\delta$.
    Hence $d_G(u,w) \le d_G(u,u') + \delta$ and $d_G(w,v) \le d_G(v',v) + \delta$, which yields the desired inequality.
\end{proof}

We will now show how, using a tree decomposition of the cluster graph of width $\Oh((1/\eps)^2)$, to construct a distance oracle and estimate eccentricities.
We give a unified proof for both statements because 
the constructed oracle is directly applied to obtain the second result.  
The presented algorithm is arguably simpler than the recursive schemes based on shortest path separators since in the analysis we only need to manage $\Oh(1)$ rounding errors, rather than $\Oh(\log n)$.
We no longer need to assume that the excluded minor is an apex graph and this will allow us later to apply \Cref{thm:preprocessing} in a more general setting.

\begin{theorem}\label{thm:preprocessing}
    Suppose we are given an edge-weighted $n$-vertex $K_h$-minor-free graph $G$, $\delta > 0$, a $\delta$-radius clustering $\ccal$ of $G$, and a binary rooted tree decomposition $(T, \chi)$  of $G_{\cal C}$ of width~$w$, depth $\Oh(\log n)$, and size $\Oh(n)$.
    \begin{enumerate}
        \item In time $\Oh_h(w \cdot n \log n)$
        we can build a data structure that, when queried for $u, v \in V(G)$, outputs $x \in [d_G(u,v),\, d_G(u,v)+2\cdot\delta]$ in time $\Oh(w \cdot \log n)$.\label{item:thm:oracle}
        
        \item Assume additionally that $G$ is connected and let  $\tau = \diam(G) / \delta$. In time $\Oh_h(w^{h+1}\cdot \tau^{h} \cdot n \log^2 n)$ we can compute eccentricity of each vertex $v \in V(G)$ within additive error $16\cdot\delta$. Moreover, within the same time we can also output a vertex $u \in V(G)$ for which $d_G(v,u) \ge \ecc_G(v) - 16\cdot\delta$.\label{item:thm:ecce}
    \end{enumerate}
\end{theorem}

Let us clarify that the parameter $\tau$ does not have to be given and it only appears in the analysis.
However, we can always assume that the diameter is known within factor 2 after computing any eccentricity.

\begin{proof}
    We begin with some notation that will be used throughout the proof.
    We will use variables $u,v$ to denote vertices of $G$ and variables $i,j$ to index the set of clusters. 
    For a cluster $i \in \ccal$, let $\loc(i)$ denote the node $t \in V(T)$ that is closest to the root among those satisfying $i \in \chi(t)$.
    We naturally extend this notation to vertices of $G$:
    all vertices $v \in C_i$ receive $\loc(v) = \loc(i)$.
    This mapping can be computed for all clusters and vertices in linear time by scanning $T$ top-down.

    For $t \in V(T)$ let $U_t \sub \ccal$ denote the set of clusters $i \in \ccal$ for which $\loc(i)$ lies in the subtree of $T$ rooted at $t$, inclusively.
    Next, let $V_t = \bigcup_{i \in U_t} C_i$.
    

    \begin{claim}\label{claim:sum-vt}
        It holds that $\sum_{t \in V(T)} |V_t| = \Oh(n\log n)$, and all the sets $V_t$ can be computed in total time $\Oh(n\log n)$.
    \end{claim}
    \begin{claimproof}
        For each $i \in \ccal$ the nodes $t$ satisfying $i \in U_t$ must be lie on the $\loc(i)$-to-root path in $T$ and so there are $\Oh(\log n)$ of them.
        Hence each $i \in \ccal$ contributes $\Oh(|C_i| \cdot \log n)$ to the sum $\sum_{t \in V(T)} |V_t|$, which yields the upper bound as $\ccal$ is a partition of $|V(G)|$. This argument leads directly to an algorithm computing the sets $V_t$.

    \end{claimproof}


\subparagraph*{Distance oracle.}
    Consider $t \in V(T)$, $i \in \cal C$, and $v \in V(G)$.
    We call $(t,i,v)$ a {\em valid triple} when $t = \loc(i)$ and  $v \in V_t$.
    For fixed $t\in V(T)$, the number of valid triples in which $t$ appears is at most $(w+1) \cdot |V_t|$, and so by \Cref{claim:sum-vt}
    the number of valid triples is $\Oh(w\cdot n\log n)$.

    Recall that for $i \in \ccal$ the vertex $c_i \in C_i$ is the center of cluster $i$.
    We will populate table $\inner$ indexed by valid triples, so that $\inner[t,i,v]$ will store the exact $(v,c_i)$-distance within the graph $G[V_t]$.
    If there is no such path, we  set $\inner[t,i,v] = \infty$.
    
    \begin{claim}\label{claim:inner-compute}
        We can compute $\inner[t,i,v]$ for all valid triples
        in total time $\Oh_h(w\cdot n\log n)$.
    \end{claim}
    \begin{claimproof}
        Consider $t \in V(T)$.
        For each $i \in \ccal$ with $\log(i) = t$, we compute all distances from $c_i$ in $G[V_t]$ using the linear single-source shortest path algorithm for $K_h$-minor-free graphs~\cite{TazariM09}.
        The total running time, over all $t \in V(T)$, is proportional to the number of valid triples.
    \end{claimproof}

    We now argue that information gathered in table $\inner$ is sufficient to retrieve all distances in $G$, up to small additive error.
    This is the only place where we rely on the fact that each cluster has small {\em strong} diameter.

\begin{figure}
\centering\includegraphics[width=12cm]{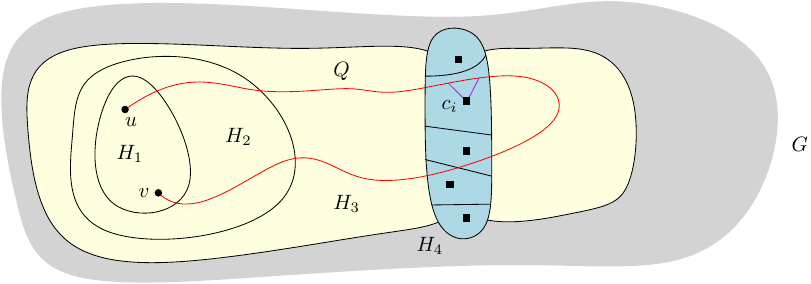}
\caption{Illustration to \Cref{claim:inner-path}.
By inspecting the nodes $t \in V(T)$ that are ancestral to both $\loc(u),\loc(v)$, we obtain a family of subgraphs $H_1 \subset H_2 \subset H_3 \subset H_4 \subset H_5 = G$, induced by the sets $V_{t}$.
Here, $H_4$ is the smallest among them that accommodates the shortest $(u,v)$-path $Q$. This path must intersect some cluster $C_i$ for which $\loc(i) = t$, where $G[V_t] = H_4$, as otherwise it would be contained in $H_3$.
Hence the $(u,v)$-distance in $H_4$ is the same as in $G$ and we can stretch $Q$ within $H_4$ to go through $c_i$ using \Cref{lem:path-touch}.
}
\label{fig:oracle}
\end{figure}

    \begin{claim}\label{claim:inner-path}
        For each $u, v \in V(G)$
        there exist $t \in V(T)$ and $i \in \ccal$ such that:
        \begin{enumerate}
            \item $(t,i,u)$ and $(t,i,v)$ are valid triples, \label{item:inner-path:valid}
            \item $\inner[t,i,u] + \inner[t,i,v] \le d_G(u,v) + 2\delta$. \label{item:inner-path:dist}
        \end{enumerate}
    \end{claim}
    \begin{claimproof}
        Let $Q$ be a shortest $(u,v)$-path in $G$.
        Next, let $\ccal_Q \sub \ccal$ denote the set of clusters intersected by $Q$ and let $T_Q \sub V(T)$ be the set of nodes in which $\ccal_Q$ appear, i.e., $T_Q = \chi^{-1}(\ccal_Q)$.
        Since $G_\ccal$ is a contraction of $G$, the subgraph $G_\ccal[\ccal_Q]$ is connected.
        Consequently, $T_Q$ induces a connected subtree of $T$.
        Let $t \in T_Q$ be the node that is closest to the root among $T_Q$.
        Note that $t = \loc(i)$ for some $i \in \ccal_Q$.

        We claim that the pair $(t,i)$ satisfies the requested conditions.
        Note that for each $j \in \ccal_Q$ the node $\loc(j)$ lies in the subtree of $T$ rooted at $t$.
        Hence $\ccal_Q \sub U_t$ and so $V(Q) \sub V_t$.
        In particular, we obtain that $u,v \in V_t$ 
        so $(t,i,u)$, $(t,i,v)$ are valid triples (\Cref{item:inner-path:valid}).

        Since $Q$ is a shortest $(u,v)$-path and $V(Q) \sub V_t$, we infer that the $(u,v)$-distance in $G[V_t]$ equals their distance in $G$.
        Now we apply \Cref{lem:path-touch} to the graph $H = G[V_t]$ and $C = C_i,\, w = c_i$.
        Because $i \in \ccal_Q$, we know that $Q$ intersects $C_i$.
        It is crucial that all $C_i$-vertices are within distance $\delta$ from $c_i$ when considering distances within $G[C_i]$ (\Cref{def:clustering}(\ref{item:def:strong})), not just in $G$.
        As $C_i \sub V_t$, this property carries over to the graph $H$.
        \Cref{lem:path-touch} implies that $d_H(u,c_i) + d_H(c_i,v) \le d_H(u,v) + 2\delta$.
        We have $\inner(t,i,u) = d_H(u,c_i)$, $\inner(t,i,v) = d_H(c_i,v)$, and $d_H(u,v) = d_G(u,v)$, which yields \Cref{item:inner-path:dist}.     
    \end{claimproof}

    \begin{claim}[Oracle]\label{claim:oracle}
        Let $u, v \in V(G)$ be given. 
        Using the table $\inner$ we can output $x \in [d_G(u,v),\, d_G(u,v) + 2\delta]$ in time
         $\Oh(w\cdot \log n)$.
    \end{claim}
    \begin{claimproof}
        We iterate over all pairs $t\in V(T),\, i \in \ccal$ for which $(t,i,u)$ and $(t,i,v)$ are valid triples. 
        There are only $\Oh(\log n)$ candidates for $t$ just because $(t,i,u)$ is a valid triple and so $t$ must lie on the $\loc(u)$-to-root path.
        For each such $t$ there are clearly at most $w+1$ candidates for $i$.

        We return $x$ equal to the minimum value of $\inner[t,i,u] + \inner[t,i,v]$ taken over all such pairs $(t,i)$.
        For each $(t,i)$ the considered sum is either $\infty$ or it equals the length of some $(u,v)$-walk in $G[V_t]$.
        This implies that $x \ge d_G(u,v)$.
        The inequality $x \le d_G(u,v) + 2\delta$ follows from
        \Cref{claim:inner-path}.
    \end{claimproof}

Claims \ref{claim:inner-compute} and \ref{claim:oracle} yield the additive distance oracle (\Cref{item:thm:oracle}).
Now we move on to a proof sketch for \Cref{item:thm:ecce}, whose full proof is postponed to \Cref{app:ecce}.

\begin{figure}
\centering\includegraphics[width=12cm]{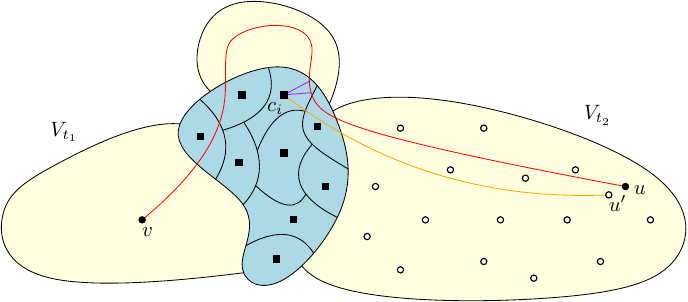}
\caption{Illustration to \Cref{claim:corset-analysis}.
The blue area corresponds to the union of clusters $C_i$ over $i \in \chi(t)$.
For any $v \in V_{t_1}$, $u \in V_{t_2}$ the shortest $(v,u)$-path (red) must traverse some $C_i$ so by \Cref{lem:path-touch} there is a slightly longer path that goes through $c_i$ (purple).
The hollow discs represent the additive core-set  of $V_{t_2}$ with respect to the centers from $\chi(t)$.
This core-set contains a vertex $u'$ which is comparable to $u$ with respect to distances from  $\{c_i \colon i \in \chi(t)\}$ (e.g., the orange path).
Hence $u'$ can simulate $u$ for the task of estimating the eccentricity of~$v$.
}
\label{fig:coreset}
\end{figure}

\subparagraph*{Eccentricities.} 
For a non-root $t \in V(T)$ we refer to its parent in $T$ as $\parent(t)$.
For $u,v \in V(G)$ let $\apxdist(u,v)$ denote the value returned by the oracle from \Cref{claim:oracle} when queried for the pair $u,v$. 
This is a well-defined function because the oracle is deterministic.
Recall the concept of an additive core-set from \Cref{lem:corest-exist} and that $\tau = \diam(G) / \delta$.

\begin{restatable}[$\bigstar$]{claim}{claimCoreset}
\label{claim:corset-compute}
        In total time $\Oh\br{(w\cdot\tau)^h \cdot n + w^2 \cdot n\log^2 n}$
        we can compute for each non-root $t \in V(T)$ a~subset $V'_t \sub V_t$ of size $\Oh\br{w^{h-1}\tau^{h}}$ that is a $(10\cdot\delta)$-additive core-set of $V_t$ with respect to $\chi(\parent(t))$.
\end{restatable}

The total number of distances we care about is $(w+1) \cdot \sum_{t\in V(T)} |V_t| = \Oh(w \cdot n\log n)$.
Each of them can be estimated using the oracle in time $\Oh(w \log n)$.
We would like to use those numbers to compute the additive core-sets, however a priori the bound from \Cref{lem:corest-exist} might not hold for approximate distances.
Nevertheless, we prove that a greedy procedure still constructs a small core-set, even after perturbation of distances.
A similar issue is handled in \cite{li2019planar} via randomization but we can afford a simpler argument because our strategy allows us to first query the oracle for all the distances between $V_t$ and $\chi(\parent(t))$.

In the next claim, we argue that the core-sets suffice to estimate the maximal distance between $v$ and $u$, where $v$ is fixed and $u$ is quantified over all vertices from some subtree. 
    There are three sources of accuracy loss: relying on the additive core-set, on the approximate distance oracle, and stretching the paths along a center from $\chi(t)$ (\Cref{lem:path-touch}).

    \begin{restatable}[$\bigstar$]{claim}{claimAnal}
    \label{claim:corset-analysis}
        Consider $t \in V(T)$ with children $t_1, t_2$.
        Let $v \in V_{t_1}$ and $V' \sub V_{t_2}$ be a $\rho$-additive core-set with respect to $\chi(t)$.
        Let $y = \max d_G(v,u)$ over all $u \in V_{t_2}$, and $x$ be defined as follows:
        
        \[x = \max_{u \in V'}  \min_{i \in \chi(t)} (\apxdist(v,c_i) +  \apxdist(c_i,u)).\]
        Then $|x - y|
        \le \rho + 6\cdot \delta$.
        Furthermore, the vertex $u^* \in V'$ realizing $x$ satisfies $d_G(v,u^*) \ge y - \rho - 6\cdot \delta$.
    \end{restatable}

We now use the additive core-sets 
 to estimate the eccentricity of a vertex.
 Essentially, we apply \Cref{claim:corset-analysis} for each choice of the lowest common ancestor of $v,u$; hence $\Oh(\log n)$ times.

    \begin{restatable}[$\bigstar$]{claim}{claimEcce}\label{claim:ecce}
        For a given vertex $v \in V(G)$ we can compute its eccentricity within additive error $16\cdot\delta$ in time $\Oh\br{w^{h+1}\cdot\tau^{h} \cdot \log^2 n}$.
        Within the same time, we can output a vertex $u$ for which $d_G(v,u) \ge \ecc_G(v) - 16\cdot\delta$.
    \end{restatable}
    
\Cref{thm:preprocessing}(\ref{item:thm:ecce}) follows by applying \Cref{claim:ecce} to each vertex.
\end{proof}

\thmFirst*
\begin{proof}
    One can assume that $G$ is connected as otherwise the diameter is $\infty$.
    We compute the eccentricity $D$ of an arbitrary vertex $v$: this can be done in linear time~\cite{TazariM09}.
    The diameter of $G$ must be contained in $[D, \, 2D]$.
    We apply \Cref{lem:local-tw} for $\delta = \eps D/16$  in time $\Oh(n \log n)$; let $\ccal$ denote the obtained $\delta$-radius clustering. 
    Next, let $\tau = \diam(G) / \delta$; note that $\tau = \Oh(1/\eps)$.
    By \Cref{lem:local-tw}(\ref{item:local:tw}) we know that $\tw(G_\ccal) = \Oh_h(1/\eps)$.
    We use \Cref{lem:prelim:twd-compute} to build a binary tree decomposition of $\tw(G_\ccal)$ of width $w = \Oh_h(\eps^{-2})$, depth $\Oh(\log n)$ and size $\Oh(n)$; this takes time $\Oh_h(\eps^{-7}\cdot n\log n)$.
    Now we can take advantage of \Cref{thm:preprocessing}(\ref{item:thm:ecce}) to compute the eccentricities of all vertices within additive error $16\cdot\delta \le \eps\cdot\diam(G)$, together with vertices that realize these distance approximately.
    This takes time $\Oh_h(w^{h+1} \cdot \tau^{h} \cdot n \log^2 n) = \Oh_h((1/\eps)^{3h+2} \cdot n \log^2 n)$.
    Finally, observe that $h$ is greater than 1 in any non-trivial instance so $7 <  3h+2$ and so the time needed to compute the tree decomposition is negligible. 
\end{proof}

In \Cref{app:main} we show how to replace additive error in the distance oracle with multiplicative one, which yields \Cref{thm:oracle}.

\section{Handling Apices}
\label{sec:apices}

We move on to the proof of \Cref{thm:apex}, extending the diameter approximation to graphs $G$ with an apex set $A$, for which $G-A$ is apex-minor-free.
We do not formulate an analogous theorem 
generalizing the distance oracle because it is trivial.
To see this, first precompute all the distances from each $w \in A$ and then, given $u,v$, return the minimum of their distance in $G-A$ (which can be handled with \Cref{lem:oracle-delta}) and $\min_{w \in A} d(u,w) + d(w,v)$. 

The reason why we cannot apply the same reduction for diameter computation is that $\diam(G-A)$ might be huge compared to $\diam(G)$.
Hence we lose the relation between the diameter and radii of clusters (governing the additive error), crucial to bound the treewidth of the cluster graph.
As a remedy, we will first preprocess the graph $G-A$ to chop it into connected components of bounded diameter.

\begin{lemma}\label{lem:diam-reduction}
    Let $H$ be a fixed graph and $c \in \nn$.
    Let $G$ be a connected edge-weighted $n$-vertex connected graph, $A \sub V(G)$  be such that $G-A$ is $H$-minor-free and $|A|=c$, and $D \in \rr$ satisfy $2D \ge \diam(G)$.
    There is an algorithm that, given the above and $\eps > 0$, runs in time $\Oh_{H,c}((1/\eps)^{c} \cdot n)$ and either correctly reports that $\diam(G) > D$ or outputs a connected edge-weighted $n$-vertex graph $G'$ and  $A' \sub V(G')$ of size $c$, with the following guarantees.
    \newline 
    
    \begin{enumerate}
        \item $G' - A'$ is $H$-minor-free.
        \item $\diam(G') \ge \diam(G)$.
        \item If $\diam(G) \le D$ then $\diam(G') \le (1+\eps)\cdot D$.
        \item Each connected component of $G' - A'$ has strong diameter at most $8 \cdot (2/\eps)^{c} \cdot D$.
    \end{enumerate}
\end{lemma}
\begin{proof}
    First of all, we remove all edges with weight larger then $2D$ as they cannot contribute to any shortest path.
    Observe that $G$ is $K_h$-minor-free for $h = |V(H)| + c$. 
    We compute shortest paths from each vertex $w\in A$ using the linear SSSP algorithm for minor-free graphs~\cite{TazariM09}.
    If any computed distance is greater than $D$, we can terminate the algorithm.
    Otherwise for each $w \in A$, $v \in V(G)$, we insert the $wv$-edge with weight $d_G(w,v)$ or update the weight if the edge is already present.
    This preserves all the pairwise distances in $G$ and does not affect the graph $G-A$.
    We will keep the variable $G$ to refer to the modified graph.
    The performed modification ensures the following property.

    \begin{claim}\label{claim:internal}
        Let $u,v \in V(G)$ and suppose that there is a shortest $(u,v)$-path in $G$ that intersects $A$.
        Then there is a shortest $(u,v)$-path with one internal vertex that belongs to~$A$.
    \end{claim}

    For $u,v \in V(G)$ let $d'(u,v)$ denote the largest integer $j$ for which $j\cdot \eps D/2 \le d_G(u,v)$.
    For $v \in V(G) \sm A$ we define its {\em signature} as $S_v = \br{d'(v,a_1),\dots,d'(v,a_c)}$ where $(a_1,\dots, a_c)$ is some fixed ordering of $A$.
    Let $\scal = \{S_v \colon v \in V(G) \sm A\}$.
    Since we have not terminated the algorithm in the first step, each $d'(v,a_i)$ is most $\lfloor 2/\eps \rfloor$ and so $|\scal| \le (2/\eps)^{c}$.

    We say that a pair $S_1,S_2 \in \scal$ is {\em relevant} if
    for each $i \in [c]$ we have $S_1[i] + S_2[i] > 2/\eps$.
    We also call $S \in \scal$ relevant if it appears in at least one relevant pair.

    \begin{claim}\label{claim:relevant}
        Let $u,v \in V(G) \sm A$ be such that $(S_u,S_v)$ is a relevant pair.
        Then there is no $(u,v)$-path of length $\le D$ that goes through $A$.
    \end{claim}
    \begin{claimproof}
        If there was such a path then $d_G(u,w) + d_G(w,v) \le D$ for some $w \in A$.
        But for each $w \in A$ we have $d_G(u,w) + d_G(w,v) \ge (\eps D/2) \cdot (d'(u,w) + d'(w,v)) > D$.
    \end{claimproof}

    \begin{claim}\label{claim:strong-radius}
        Let $u,v \in V(G) \sm A$ be such that $S_u = S_v$ and $S_u$ is relevant.
        If $d_{G-A}(u,v) > 2D$ then $\diam(G) > D$.
    \end{claim}
    \begin{claimproof}
        Suppose that $\diam(G) \le D$.
        Since $S_u$ is relevant, there exists $z \in V(G) \sm A$ such that $(S_u,S_z)$ is a relevant pair.
        Due to \Cref{claim:relevant}, the shortest $(u,z)$-path avoids $A$ so $d_{G-A}(u,z) \le D$.
        By the same argument  $d_{G-A}(v,z) \le D$ and the claim follows from the triangle inequality.
    \end{claimproof}

    For each relevant $S \in \scal$ let us choose an arbitrary $v$ with $S_v = S$; we denote the set of chosen vertices as $B$.
    Clearly $|B| \le |\scal| \le (2/\eps)^{c}$.
    For each $b \in B$ we compute all the distances from $b$ in $G-A$ in linear time and check whether there is any $v \in V(G) \sm A$ with $S_v = S_b$ and $d_{G-A}(b,v) > 2D$.
    If so, then by \Cref{claim:strong-radius} we can again terminate the algorithm.

\begin{figure}
\centering\includegraphics[width=14cm]{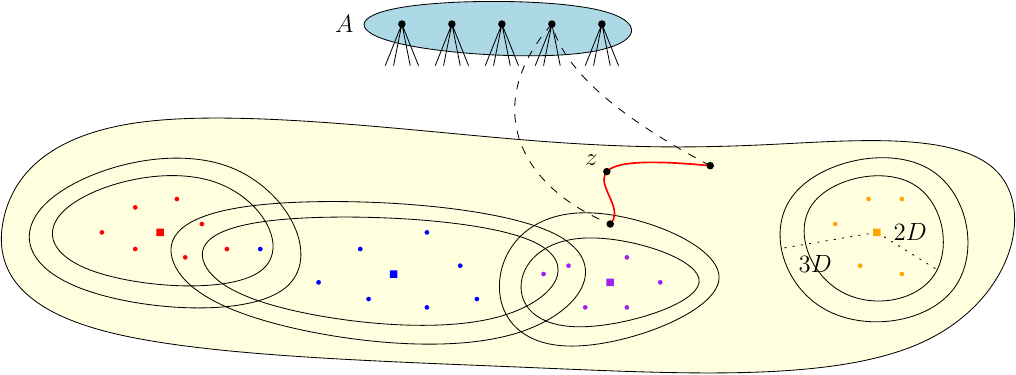}
\caption{Illustration to \Cref{claim:remove}.
The set $B$ of representatives for relevant signatures is given by the squares.
By \Cref{claim:strong-radius}, every vertex sharing the  same signature  as $b \in B$ (represented by the same color) must lie in the $2D$-radius ball around $b$ in $G-A$. If $z$ is outside all $3D$-radius balls, then every path in $G-A$ of length $\le D$ traversing $z$ must connect vertices $u,v$ for which $(S_u,S_v)$ is an irrelevant pair.
For such pairs, there must be an almost shortest path traversing $A$.
}
\label{fig:irrelevant}
\end{figure}

    \begin{claim}\label{claim:remove}
        Suppose that the algorithm has not been terminated and let $Z \sub V(G) \sm A$ be the set of vertices $z$ satisfying $d_{G-A}(z,b) > 3D$ for all $b \in B$.
        Next, let $F \sub E(G-A)$ be the set of edges from $G-A$ with at least one endpoint in $Z$.
        Suppose that $\diam(G) \le D$.
        Then $\diam(G \sm F) \le (1+\eps)\cdot D$.
    \end{claim}
    \begin{claimproof}
        Consider $u, v \in V(G)$ and let $P$ be a shortest $(u,v)$-path in $G$.
        By assumption, $P$ is not longer than $D$.
        We need to show that there is a $(u,v)$-path of length $\le (1+\eps)\cdot D$ that avoids~$F$.
        If $P$ intersects $A$ then by \Cref{claim:internal} we can assume that $P$ consists only of edges incident to $A$ (at most two).
        Then clearly $P$ avoids $F$.

        Suppose now that $V(P) \cap A = \emptyset$ and that $P$ uses some edge from $F$.
        Then there is $z \in V(P) \cap Z$ whose distance in $G-A$ to each of $u,v$ is at most $D$.
        We claim that $S_u, S_v$ are irrelevant.
        Assume w.l.o.g. that $S_u$ is relevant and let $b \in B$ be the chosen vertex with $S_u = S_b$.
        We have $d_{G-A}(u,b) \le 2D$ as otherwise we would have terminated the algorithm due to \Cref{claim:strong-radius}.
        But $d_{G-A}(u,z) \le D$ so $d_{G-A}(z,b) \le 3D$.
        This contradicts the property used to define $Z$.
        
        We established that $S_u, S_v$ are irrelevant, so the pair $(S_u, S_v)$ is also irrelevant.
        By definition, this means that there is $w \in A$ such that $d'(u,w) + d'(w,v) \le 2/\eps$.
        We have  $d_G(u,w) \le (\eps D/2) \cdot (d'(u,w) + 1)$ and likewise for the pair $(v,w)$.
        We estimate 
        \[d_G(u,w) + d_G(w,v) \le (\eps D/2) \cdot (d'(u,w) + d'(w,v) + 2) \le D + (\eps D/2) \cdot 2 = (1+\eps)\cdot D.\]
        Hence the path $(u,w,v)$  is sufficiently short and it avoids $F$.
        The claim follows.
    \end{claimproof}

    We showed that removing all edges from $F$ is safe for approximating the diameter of $G$.
    It remains to prove that this operation splits $G-A$ into pieces of bounded strong diameter.

    \begin{claim}\label{claim:cc-diam}
        Let $C \sub V(G) \sm A$ be a set of vertices inducing a connected component of $(G - A) \sm F$.
        Then $\diam(G[C]) \le 8 \cdot (2/\eps)^c \cdot D$.
    \end{claim}
    \begin{claimproof}
        For $b \in B$ let $C_b$ be the set of vertices within distance $3D$ from $b$ in $G-A$.
        Observe that $C$ is either a single vertex from $Z$ or a union of $C_b$ for $b \in B'$ for some $B' \sub B$.
        In the first case the diameter is 0 so let us focus on the second case.
        
        Let $u, v \in C$ and $P$ be a shortest $(u,v)$-path in $G[C]$.
        Then $u \in C_{b_1}$ for some $b_1 \in B'$.
        Let $v_1$ be the farthest vertex on $P$ (counting from $u$) that belongs to $C_{b_1}$.
        The $(u,v_1)$-distance in $G[C]$ is at most their distance in $G[C_{b_1}]$ so $d_P(u,v_1) \le 6D$ as $P$ is a shortest path.
        Next, take $u_2$ to be the next vertex on $P$ after $v_1$ and pick $b_2 \in B$ satisfying $u_2 \in C_{b_2}$.
        Analogously as before, we find  the  vertex $v_2 \in C_{b_2}$ that is farthest on $P$ and continue this process.
        By construction, every vertex $b \in B'$ appears at most once as $b_i$ so this splits $P$ into at most $|B'|$ subpaths of length $\le 6D$.
        Also for each $i$ the edge between $v_i$ and $u_{i+1}$ has weight at most $2D$ because we have removed all the heavier edges at the very beginning.
        The claim follows from $|B'| \le |B| \le (2/\eps)^c$.  
    \end{claimproof}

    We check that $G' = G \sm F$ and $A' = A$ satisfy all the points of the lemma.
    First, $G' - A$ is a subgraph of $G - A$ so it remains $H$-minor-free.
    Next, removing edges can only increase the diameter so $\diam(G') \ge \diam(G)$ while the second inequality follows from \Cref{claim:remove}.
    The diameter bound for components of $G'-A'$ has been given in \Cref{claim:cc-diam}.
    Finally, the sets $Z$ and $F$ can be identified in time $\Oh_H(|B| \cdot n) = \Oh_{H,c}((1/\eps)^{c} \cdot n)$. 
\end{proof}

We briefly explain how \Cref{lem:diam-reduction} leads to \Cref{thm:apex}.
Given  $D \in [\diam(G),\, 2\cdot\diam(G)]$ we want to learn that $\diam(G) > D$ or that $\diam(G) \le  (1 + \eps)\cdot D$.
Let $(G', A')$ be the output given by \Cref{lem:diam-reduction} and $\delta = \Theta(\eps D)$.
Each connected component $G^i$ of $G'- A'$ 
satisfies $\diam(G^i) = \Oh((1/\eps)^{c}\cdot D)$. 
We apply \Cref{lem:local-tw} to get a $\delta$-radius clustering $\ccal^i$ of $G^i$ so that the treewidth of the cluster graph $G^i_{\ccal^i}$ is $\Oh(\diam(G^i) / \delta) = \Oh((1/\eps)^{c+1})$.
We combine those to get a $\delta$-radius clustering $\ccal'$ of $G'$ by assigning each $w \in A'$ a singleton cluster.

Using \Cref{lem:prelim:twd-compute} we can compute a balanced tree decomposition of $G'_{\ccal'}$ of width $w = \Oh((1/\eps)^{2(c+1)})$.
       It can be now processed with \Cref{thm:preprocessing}(\ref{item:thm:ecce}): in time $\Oh(w^{h+1}(1/\eps)^h\cdot n\log^2 n)$ we can estimate $\diam(G')$ within additive error $16\cdot \delta =  \Theta(\eps D)$.
If this value is too large, we learn that  $\diam(G') > (1+\eps) \cdot D$, implying  $\diam(G) > D$.
Otherwise we infer that $\diam(G) \le \diam(G') \le (1 + \eps)\cdot D$.
Repeating this procedure for different $D$ allows us to estimate $\diam(G)$ within multiplicative error $(1+\eps)$.
The details can be found in \Cref{app:main}.

\section{Conclusion and Open Problems}
\label{sec:conclusion}

We have presented almost-linear algorithms approximating diameter that bypass topological arguments and work 
way beyond surface-embeddable graphs. 
A natural next step would be to generalize these results further, to all graph classes excluding a minor.
The reason why we could not extend \Cref{cor:planar} to all cases covered in \Cref{thm:apex} is that we need to know the deletion set to the apex-minor-free class. 
The best known algorithm finding such a set requires quadratic time~\cite{SauST22}.
To apply the same strategy for the general case of $H$-minor-free graphs, one would need to compute some form of the Robertson and Seymour's decomposition, which seems an even harder task.
However, recent advances in minor testing~\cite{KorhonenPS24} suggest that this may be possible in time $\Oh_H(n^{1+o(1)})$ (see the section ``Outlook'' there).
We have not attempted to generalize \Cref{thm:apex} to work on such decompositions because (1) this would require a tedious analysis of a balanced tree decomposition of the clique-sum decomposition, and (2) this strategy currently seems  hopeless to achieve running time $\Oh_H(\poly(1/\eps, \log n)\cdot n)$.
Is there an alternative approach that would work without the decomposition at hand? 


Another question is whether
one can improve the dependence on $n$ to $\Oh(n\log n)$ or even $\Oh(n)$? 
The current running time bound follows from the fact that we need to precompute $\Omega(n\log n)$ distances and each oracle call requires time $\Omega(\log n)$.
Can this be improved?
For example, for planar metrics there exist distance oracles with constant query time (for constant $\eps)$
albeit with preprocessing time $\Omega(n\log^3 n)$~\cite{ChanS17, LeW21oracle}.

When it comes to the dependence on $\eps$, our running time is burdened by the fact that for $k$-treewidth graphs we compute a tree decomposition of width $\Oh(k^2)$.
Unfortunately, all the known almost-linear $\Oh(1)$-approximation algorithms for treewidth require time exponential in $k$.
Since we rely on a construction from~\cite{FominLSPW18} that works for general graphs, it is tempting to improve the width guarantee therein when the input graph is $H$-minor-free.
Again, this is doable in planar graphs~\cite{KammerT16}.

\bibliography{bib}

\begin{thebibliography}{10}

\bibitem{abboud2023what}
Amir Abboud, Shay Mozes, and Oren Weimann.
\newblock {What Else Can Voronoi Diagrams Do for Diameter in Planar Graphs?}
\newblock In Inge~Li G{\o}rtz, Martin Farach-Colton, Simon~J. Puglisi, and Grzegorz Herman, editors, {\em 31st Annual European Symposium on Algorithms (ESA 2023)}, volume 274 of {\em Leibniz International Proceedings in Informatics (LIPIcs)}, pages 4:1--4:20, Dagstuhl, Germany, 2023. Schloss Dagstuhl -- Leibniz-Zentrum f{\"u}r Informatik.
\newblock \href {https://doi.org/10.4230/LIPIcs.ESA.2023.4} {\path{doi:10.4230/LIPIcs.ESA.2023.4}}.

\bibitem{AbboudWW16}
Amir Abboud, Virginia {Vassilevska Williams}, and Joshua~R. Wang.
\newblock Approximation and fixed parameter subquadratic algorithms for radius and diameter in sparse graphs.
\newblock In Robert Krauthgamer, editor, {\em Proceedings of the Twenty-Seventh Annual {ACM-SIAM} Symposium on Discrete Algorithms, {SODA} 2016, Arlington, VA, USA, January 10-12, 2016}, pages 377--391. {SIAM}, 2016.
\newblock \href {https://doi.org/10.1137/1.9781611974331.CH28} {\path{doi:10.1137/1.9781611974331.CH28}}.

\bibitem{abraham2006object}
Ittai Abraham and Cyril Gavoille.
\newblock Object location using path separators.
\newblock In {\em Proceedings of the twenty-fifth annual ACM symposium on Principles of distributed computing}, pages 188--197, 2006.
\newblock \href {https://doi.org/10.1145/1146381.1146411} {\path{doi:10.1145/1146381.1146411}}.

\bibitem{abraham2007strong}
Ittai Abraham, Cyril Gavoille, Dahlia Malkhi, and Udi Wieder.
\newblock Strong-diameter decompositions of minor free graphs.
\newblock In {\em Proceedings of the nineteenth annual ACM symposium on Parallel algorithms and architectures}, pages 16--24, 2007.
\newblock \href {https://doi.org/10.1145/1248377.1248381} {\path{doi:10.1145/1248377.1248381}}.

\bibitem{barahona1983max}
Francisco Barahona.
\newblock The max-cut problem on graphs not contractible to {$K_5$}.
\newblock {\em Operations Research Letters}, 2(3):107--111, 1983.
\newblock \href {https://doi.org/10.1016/0167-6377(83)90016-0} {\path{doi:10.1016/0167-6377(83)90016-0}}.

\bibitem{BourneufP25}
Romain Bourneuf and Marcin Pilipczuk.
\newblock Bounding \emph{{\(\epsilon\)}}-scatter dimension via metric sparsity.
\newblock In Yossi Azar and Debmalya Panigrahi, editors, {\em Proceedings of the 2025 Annual {ACM-SIAM} Symposium on Discrete Algorithms, {SODA} 2025, New Orleans, LA, USA, January 12-15, 2025}, pages 3155--3171. {SIAM}, 2025.
\newblock \href {https://doi.org/10.1137/1.9781611978322.101} {\path{doi:10.1137/1.9781611978322.101}}.

\bibitem{bousquet2015vc}
Nicolas Bousquet and St{\'e}phan Thomass{\'e}.
\newblock {VC-dimension and Erd{\H{o}}s--P{\'o}sa property}.
\newblock {\em Discrete Mathematics}, 338(12):2302--2317, 2015.

\bibitem{braverman2021coresets}
Vladimir Braverman, Shaofeng H-C Jiang, Robert Krauthgamer, and Xuan Wu.
\newblock Coresets for clustering in excluded-minor graphs and beyond.
\newblock In {\em Proceedings of the 2021 ACM-SIAM Symposium on Discrete Algorithms (SODA)}, pages 2679--2696. SIAM, 2021.

\bibitem{busch2007improved}
Costas Busch, Ryan LaFortune, and Srikanta Tirthapura.
\newblock Sparse covers for planar graphs and graphs that exclude a fixed minor.
\newblock {\em Algorithmica}, 69:658--684, 2014.

\bibitem{Cabello19}
Sergio Cabello.
\newblock Subquadratic algorithms for the diameter and the sum of pairwise distances in planar graphs.
\newblock {\em {ACM} Trans. Algorithms}, 15(2):21:1--21:38, 2019.
\newblock \href {https://doi.org/10.1145/3218821} {\path{doi:10.1145/3218821}}.

\bibitem{ChakrabortyG23}
Dibyayan Chakraborty and Kshitij Gajjar.
\newblock Finding geometric representations of apex graphs is {NP-hard}.
\newblock {\em Theor. Comput. Sci.}, 971:114064, 2023.
\newblock \href {https://doi.org/10.1016/J.TCS.2023.114064} {\path{doi:10.1016/J.TCS.2023.114064}}.

\bibitem{ChanS17}
Timothy~M. Chan and Dimitrios Skrepetos.
\newblock {Faster Approximate Diameter and Distance Oracles in Planar Graphs}.
\newblock In Kirk Pruhs and Christian Sohler, editors, {\em 25th Annual European Symposium on Algorithms (ESA 2017)}, volume~87 of {\em Leibniz International Proceedings in Informatics (LIPIcs)}, pages 25:1--25:13, Dagstuhl, Germany, 2017. Schloss Dagstuhl -- Leibniz-Zentrum f{\"u}r Informatik.
\newblock \href {https://doi.org/10.4230/LIPIcs.ESA.2017.25} {\path{doi:10.4230/LIPIcs.ESA.2017.25}}.

\bibitem{ChanSkrepetos2019}
Timothy~M. Chan and Dimitrios Skrepetos.
\newblock Approximate shortest paths and distance oracles in weighted unit-disk graphs.
\newblock {\em Journal of Computational Geometry}, 10(2):3–20, Feb. 2019.
\newblock \href {https://doi.org/10.20382/jocg.v10i2a2} {\path{doi:10.20382/jocg.v10i2a2}}.

\bibitem{chang2023covering}
Hsien-Chih Chang, Jonathan Conroy, Hung Le, Lazar Milenkovic, Shay Solomon, and Cuong Than.
\newblock Covering planar metrics (and beyond): {$O(1)$} trees suffice.
\newblock In {\em 2023 IEEE 64th Annual Symposium on Foundations of Computer Science (FOCS)}, pages 2231--2261. IEEE, 2023.

\bibitem{chang2023resolving}
Hsien-Chih Chang, Jonathan Conroy, Hung Le, Lazar Milenkovic, Shay Solomon, and Cuong Than.
\newblock Resolving the steiner point removal problem in planar graphs via shortcut partitions.
\newblock {\em arXiv preprint arXiv:2306.06235}, 2023.

\bibitem{chang2024shortcut}
Hsien-Chih Chang, Jonathan Conroy, Hung Le, Lazar Milenkovi{\'c}, Shay Solomon, and Cuong Than.
\newblock Shortcut partitions in minor-free graphs: Steiner point removal, distance oracles, tree covers, and more.
\newblock In {\em Proceedings of the 2024 Annual ACM-SIAM Symposium on Discrete Algorithms (SODA)}, pages 5300--5331. SIAM, 2024.

\bibitem{chang24}
Hsien-Chih Chang, Jie Gao, and Hung Le.
\newblock {Computing Diameter+2 in Truly-Subquadratic Time for Unit-Disk Graphs}.
\newblock In Wolfgang Mulzer and Jeff~M. Phillips, editors, {\em 40th International Symposium on Computational Geometry (SoCG 2024)}, volume 293 of {\em Leibniz International Proceedings in Informatics (LIPIcs)}, pages 38:1--38:14, Dagstuhl, Germany, 2024. Schloss Dagstuhl -- Leibniz-Zentrum f{\"u}r Informatik.
\newblock \href {https://doi.org/10.4230/LIPIcs.SoCG.2024.38} {\path{doi:10.4230/LIPIcs.SoCG.2024.38}}.

\bibitem{ChangKT22}
Hsien{-}Chih Chang, Robert Krauthgamer, and Zihan Tan.
\newblock Almost-linear \emph{{\(\epsilon\)}}-emulators for planar graphs.
\newblock In Stefano Leonardi and Anupam Gupta, editors, {\em {STOC} '22: 54th Annual {ACM} {SIGACT} Symposium on Theory of Computing, Rome, Italy, June 20 - 24, 2022}, pages 1311--1324. {ACM}, 2022.
\newblock \href {https://doi.org/10.1145/3519935.3519998} {\path{doi:10.1145/3519935.3519998}}.

\bibitem{charalampopoulos2023almost}
Panagiotis Charalampopoulos, Pawe{\l} Gawrychowski, Yaowei Long, Shay Mozes, Seth Pettie, Oren Weimann, and Christian Wulff-Nilsen.
\newblock Almost optimal exact distance oracles for planar graphs.
\newblock {\em Journal of the ACM}, 70(2):1--50, 2023.

\bibitem{chatterjee2014optimal}
Krishnendu Chatterjee, Rasmus Ibsen-Jensen, and Andreas Pavlogiannis.
\newblock Optimal tree-decomposition balancing and reachability on low treewidth graphs.
\newblock {\em IST Austria}, 2014.
\newblock \href {https://doi.org/10.15479/AT:IST-2014-314-v1-1} {\path{doi:10.15479/AT:IST-2014-314-v1-1}}.

\bibitem{ChepoiEV07}
Victor Chepoi, Bertrand Estellon, and Yann Vax{\`{e}}s.
\newblock Covering planar graphs with a fixed number of balls.
\newblock {\em Discret. Comput. Geom.}, 37(2):237--244, 2007.
\newblock \href {https://doi.org/10.1007/S00454-006-1260-0} {\path{doi:10.1007/S00454-006-1260-0}}.

\bibitem{Cohen-AddadD0SS25}
Vincent Cohen{-}Addad, Andrew Draganov, Matteo Russo, David Saulpic, and Chris Schwiegelshohn.
\newblock A tight {VC}-dimension analysis of clustering coresets with applications.
\newblock In Yossi Azar and Debmalya Panigrahi, editors, {\em Proceedings of the 2025 Annual {ACM-SIAM} Symposium on Discrete Algorithms, {SODA} 2025, New Orleans, LA, USA, January 12-15, 2025}, pages 4783--4808. {SIAM}, 2025.
\newblock \href {https://doi.org/10.1137/1.9781611978322.162} {\path{doi:10.1137/1.9781611978322.162}}.

\bibitem{CohenAddadLPP23}
Vincent Cohen{-}Addad, Hung Le, Marcin Pilipczuk, and Michal Pilipczuk.
\newblock Planar and minor-free metrics embed into metrics of polylogarithmic treewidth with expected multiplicative distortion arbitrarily close to 1.
\newblock In {\em 64th {IEEE} Annual Symposium on Foundations of Computer Science, {FOCS} 2023, Santa Cruz, CA, USA, November 6-9, 2023}, pages 2262--2277. {IEEE}, 2023.
\newblock \href {https://doi.org/10.1109/FOCS57990.2023.00140} {\path{doi:10.1109/FOCS57990.2023.00140}}.

\bibitem{DemaineH04}
Erik~D. Demaine and Mohammad~Taghi Hajiaghayi.
\newblock Equivalence of local treewidth and linear local treewidth and its algorithmic applications.
\newblock In J.~Ian Munro, editor, {\em Proceedings of the Fifteenth Annual {ACM-SIAM} Symposium on Discrete Algorithms, {SODA} 2004, New Orleans, Louisiana, USA, January 11-14, 2004}, pages 840--849. {SIAM}, 2004.
\newblock URL: \url{https://dl.acm.org/doi/10.5555/982792.982919}.

\bibitem{DemaineHK09}
Erik~D. Demaine, MohammadTaghi Hajiaghayi, and Ken{-}ichi Kawarabayashi.
\newblock Approximation algorithms via structural results for apex-minor-free graphs.
\newblock In Susanne Albers, Alberto Marchetti{-}Spaccamela, Yossi Matias, Sotiris~E. Nikoletseas, and Wolfgang Thomas, editors, {\em Automata, Languages and Programming, 36th International Colloquium, {ICALP} 2009, Rhodes, Greece, July 5-12, 2009, Proceedings, Part {I}}, volume 5555 of {\em Lecture Notes in Computer Science}, pages 316--327. Springer, 2009.
\newblock \href {https://doi.org/10.1007/978-3-642-02927-1\_27} {\path{doi:10.1007/978-3-642-02927-1\_27}}.

\bibitem{ducoffe2022diameter}
Guillaume Ducoffe, Michel Habib, and Laurent Viennot.
\newblock {Diameter, eccentricities and distance oracle computations on H-minor free graphs and graphs of bounded (distance) Vapnik--Chervonenkis dimension}.
\newblock {\em SIAM Journal on Computing}, 51(5):1506--1534, 2022.

\bibitem{duraj2023better}
Lech Duraj, Filip Konieczny, and Krzysztof Potepa.
\newblock Better diameter algorithms for bounded vc-dimension graphs and geometric intersection graphs.
\newblock In Timothy~M. Chan, Johannes Fischer, John Iacono, and Grzegorz Herman, editors, {\em 32nd Annual European Symposium on Algorithms, {ESA} 2024, September 2-4, 2024, Royal Holloway, London, United Kingdom}, volume 308 of {\em LIPIcs}, pages 51:1--51:18. Schloss Dagstuhl - Leibniz-Zentrum f{\"{u}}r Informatik, 2024.
\newblock \href {https://doi.org/10.4230/LIPICS.ESA.2024.51} {\path{doi:10.4230/LIPICS.ESA.2024.51}}.

\bibitem{Eppstein00}
David Eppstein.
\newblock Diameter and treewidth in minor-closed graph families.
\newblock {\em Algorithmica}, 27(3):275--291, 2000.
\newblock \href {https://doi.org/10.1007/S004530010020} {\path{doi:10.1007/S004530010020}}.

\bibitem{Filtser2004sparse}
Arnold Filtser.
\newblock On sparse covers of minor free graphs, low dimensional metric embeddings, and other applications.
\newblock {\em CoRR}, abs/2401.14060, 2024.
\newblock \href {https://arxiv.org/abs/2401.14060} {\path{arXiv:2401.14060}}.

\bibitem{Arnold22}
Arnold Filtser and Hung Le.
\newblock Low treewidth embeddings of planar and minor-free metrics.
\newblock In {\em 2022 IEEE 63rd Annual Symposium on Foundations of Computer Science (FOCS)}, pages 1081--1092, 2022.
\newblock \href {https://doi.org/10.1109/FOCS54457.2022.00105} {\path{doi:10.1109/FOCS54457.2022.00105}}.

\bibitem{FominLSPW18}
Fedor~V. Fomin, Daniel Lokshtanov, Saket Saurabh, Michal Pilipczuk, and Marcin Wrochna.
\newblock Fully polynomial-time parameterized computations for graphs and matrices of low treewidth.
\newblock {\em {ACM} Trans. Algorithms}, 14(3):34:1--34:45, 2018.
\newblock \href {https://doi.org/10.1145/3186898} {\path{doi:10.1145/3186898}}.

\bibitem{FoxEpsteinKS19}
Eli Fox{-}Epstein, Philip~N. Klein, and Aaron Schild.
\newblock Embedding planar graphs into low-treewidth graphs with applications to efficient approximation schemes for metric problems.
\newblock In Timothy~M. Chan, editor, {\em Proceedings of the Thirtieth Annual {ACM-SIAM} Symposium on Discrete Algorithms, {SODA} 2019, San Diego, California, USA, January 6-9, 2019}, pages 1069--1088. {SIAM}, 2019.
\newblock \href {https://doi.org/10.1137/1.9781611975482.66} {\path{doi:10.1137/1.9781611975482.66}}.

\bibitem{fredslundhansen}
Viktor Fredslund-Hansen, Shay Mozes, and Christian Wulff-Nilsen.
\newblock {Truly Subquadratic Exact Distance Oracles with Constant Query Time for Planar Graphs}.
\newblock In Hee-Kap Ahn and Kunihiko Sadakane, editors, {\em 32nd International Symposium on Algorithms and Computation (ISAAC 2021)}, volume 212 of {\em Leibniz International Proceedings in Informatics (LIPIcs)}, pages 25:1--25:12, Dagstuhl, Germany, 2021. Schloss Dagstuhl -- Leibniz-Zentrum f{\"u}r Informatik.
\newblock \href {https://doi.org/10.4230/LIPIcs.ISAAC.2021.25} {\path{doi:10.4230/LIPIcs.ISAAC.2021.25}}.

\bibitem{GawrychowskiKMS21}
Pawel Gawrychowski, Haim Kaplan, Shay Mozes, Micha Sharir, and Oren Weimann.
\newblock Voronoi diagrams on planar graphs, and computing the diameter in deterministic {\~{o}}(n\({}^{\mbox{5/3}}\)) time.
\newblock {\em {SIAM} J. Comput.}, 50(2):509--554, 2021.
\newblock \href {https://doi.org/10.1137/18M1193402} {\path{doi:10.1137/18M1193402}}.

\bibitem{gawrychowski2018better}
Pawel Gawrychowski, Shay Mozes, Oren Weimann, and Christian Wulff-Nilsen.
\newblock Better tradeoffs for exact distance oracles in planar graphs.
\newblock In {\em Proceedings of the Twenty-Ninth Annual ACM-SIAM Symposium on Discrete Algorithms}, pages 515--529. SIAM, 2018.

\bibitem{Husfeldt16}
Thore Husfeldt.
\newblock {Computing Graph Distances Parameterized by Treewidth and Diameter}.
\newblock In Jiong Guo and Danny Hermelin, editors, {\em 11th International Symposium on Parameterized and Exact Computation (IPEC 2016)}, volume~63 of {\em Leibniz International Proceedings in Informatics (LIPIcs)}, pages 16:1--16:11, Dagstuhl, Germany, 2017. Schloss Dagstuhl -- Leibniz-Zentrum f{\"u}r Informatik.
\newblock \href {https://doi.org/10.4230/LIPIcs.IPEC.2016.16} {\path{doi:10.4230/LIPIcs.IPEC.2016.16}}.

\bibitem{JansenLS14}
Bart M.~P. Jansen, Daniel Lokshtanov, and Saket Saurabh.
\newblock A near-optimal planarization algorithm.
\newblock In Chandra Chekuri, editor, {\em Proceedings of the Twenty-Fifth Annual {ACM-SIAM} Symposium on Discrete Algorithms, {SODA} 2014, Portland, Oregon, USA, January 5-7, 2014}, pages 1802--1811. {SIAM}, 2014.
\newblock \href {https://doi.org/10.1137/1.9781611973402.130} {\path{doi:10.1137/1.9781611973402.130}}.

\bibitem{KammerT16}
Frank Kammer and Torsten Tholey.
\newblock Approximate tree decompositions of planar graphs in linear time.
\newblock {\em Theor. Comput. Sci.}, 645:60--90, 2016.
\newblock URL: \url{https://doi.org/10.1016/j.tcs.2016.06.040}, \href {https://doi.org/10.1016/J.TCS.2016.06.040} {\path{doi:10.1016/J.TCS.2016.06.040}}.

\bibitem{KarczmarzZ25}
Adam Karczmarz and Da~Wei Zheng.
\newblock Subquadratic algorithms in minor-free digraphs: (weighted) distance oracles, decrementai reachability, and more.
\newblock In Yossi Azar and Debmalya Panigrahi, editors, {\em Proceedings of the 2025 Annual {ACM-SIAM} Symposium on Discrete Algorithms, {SODA} 2025, New Orleans, LA, USA, January 12-15, 2025}, pages 4338--4351. {SIAM}, 2025.
\newblock \href {https://doi.org/10.1137/1.9781611978322.147} {\path{doi:10.1137/1.9781611978322.147}}.

\bibitem{kawarabayashi2011linear}
Ken-ichi Kawarabayashi, Philip~N Klein, and Christian Sommer.
\newblock Linear-space approximate distance oracles for planar, bounded-genus and minor-free graphs.
\newblock In {\em International Colloquium on Automata, Languages, and Programming}, pages 135--146. Springer, 2011.

\bibitem{KPR93}
Philip Klein, Serge~A. Plotkin, and Satish Rao.
\newblock Excluded minors, network decomposition, and multicommodity flow.
\newblock In {\em Proceedings of the Twenty-Fifth Annual ACM Symposium on Theory of Computing}, STOC '93, page 682–690, New York, NY, USA, 1993. Association for Computing Machinery.
\newblock \href {https://doi.org/10.1145/167088.167261} {\path{doi:10.1145/167088.167261}}.

\bibitem{kluk2025faster}
Kacper Kluk, Marcin Pilipczuk, Micha{\l} Pilipczuk, and Giannos Stamoulis.
\newblock Faster diameter computation in graphs of bounded euler genus.
\newblock {\em arXiv preprint arXiv:2502.07501}, 2025.

\bibitem{KorhonenPS24}
Tuukka Korhonen, Michal Pilipczuk, and Giannos Stamoulis.
\newblock Minor containment and disjoint paths in almost-linear time.
\newblock In {\em 65th {IEEE} Annual Symposium on Foundations of Computer Science, {FOCS} 2024, Chicago, IL, USA, October 27-30, 2024}, pages 53--61. {IEEE}, 2024.
\newblock \href {https://doi.org/10.1109/FOCS61266.2024.00014} {\path{doi:10.1109/FOCS61266.2024.00014}}.

\bibitem{Le23}
Hung Le.
\newblock Approximate distance oracles for planar graphs with subpolynomial error dependency.
\newblock In Nikhil Bansal and Viswanath Nagarajan, editors, {\em Proceedings of the 2023 {ACM-SIAM} Symposium on Discrete Algorithms, {SODA} 2023, Florence, Italy, January 22-25, 2023}, pages 1877--1904. {SIAM}, 2023.
\newblock \href {https://doi.org/10.1137/1.9781611977554.CH72} {\path{doi:10.1137/1.9781611977554.CH72}}.

\bibitem{LeW21oracle}
Hung Le and Christian Wulff{-}Nilsen.
\newblock Optimal approximate distance oracle for planar graphs.
\newblock In {\em 62nd {IEEE} Annual Symposium on Foundations of Computer Science, {FOCS} 2021, Denver, CO, USA, February 7-10, 2022}, pages 363--374. {IEEE}, 2021.
\newblock \href {https://doi.org/10.1109/FOCS52979.2021.00044} {\path{doi:10.1109/FOCS52979.2021.00044}}.

\bibitem{LeW24}
Hung Le and Christian Wulff{-}Nilsen.
\newblock {VC} set systems in minor-free (di)graphs and applications.
\newblock In David~P. Woodruff, editor, {\em Proceedings of the 2024 {ACM-SIAM} Symposium on Discrete Algorithms, {SODA} 2024, Alexandria, VA, USA, January 7-10, 2024}, pages 5332--5360. {SIAM}, 2024.
\newblock \href {https://doi.org/10.1137/1.9781611977912.192} {\path{doi:10.1137/1.9781611977912.192}}.

\bibitem{li2019planar}
Jason Li and Merav Parter.
\newblock Planar diameter via metric compression.
\newblock In {\em Proceedings of the 51st Annual ACM SIGACT Symposium on Theory of Computing}, pages 152--163, 2019.

\bibitem{robertson2003graph}
Neil Robertson and Paul~D Seymour.
\newblock {Graph minors. XVI. Excluding a non-planar graph}.
\newblock {\em Journal of Combinatorial Theory, Series B}, 89(1):43--76, 2003.

\bibitem{Roditty13W}
Liam Roditty and Virginia Vassilevska~Williams.
\newblock Fast approximation algorithms for the diameter and radius of sparse graphs.
\newblock In {\em Proceedings of the Forty-Fifth Annual ACM Symposium on Theory of Computing}, STOC '13, page 515–524, New York, NY, USA, 2013. Association for Computing Machinery.
\newblock \href {https://doi.org/10.1145/2488608.2488673} {\path{doi:10.1145/2488608.2488673}}.

\bibitem{SauST22}
Ignasi Sau, Giannos Stamoulis, and Dimitrios~M. Thilikos.
\newblock {\emph{k}-apices of Minor-closed Graph Classes. {II.} Parameterized Algorithms}.
\newblock {\em {ACM} Trans. Algorithms}, 18(3):21:1--21:30, 2022.
\newblock \href {https://doi.org/10.1145/3519028} {\path{doi:10.1145/3519028}}.

\bibitem{sauer1972density}
Norbert Sauer.
\newblock On the density of families of sets.
\newblock {\em Journal of Combinatorial Theory, Series A}, 13(1):145--147, 1972.

\bibitem{TazariM09}
Siamak Tazari and Matthias M{\"{u}}ller{-}Hannemann.
\newblock Shortest paths in linear time on minor-closed graph classes, with an application to steiner tree approximation.
\newblock {\em Discret. Appl. Math.}, 157(4):673--684, 2009.
\newblock \href {https://doi.org/10.1016/J.DAM.2008.08.002} {\path{doi:10.1016/J.DAM.2008.08.002}}.

\bibitem{Thorup04}
Mikkel Thorup.
\newblock Compact oracles for reachability and approximate distances in planar digraphs.
\newblock {\em J. ACM}, 51(6):993–1024, November 2004.
\newblock \href {https://doi.org/10.1145/1039488.1039493} {\path{doi:10.1145/1039488.1039493}}.

\bibitem{vapnik2015uniform}
Vladimir~N Vapnik and A~Ya Chervonenkis.
\newblock On the uniform convergence of relative frequencies of events to their probabilities.
\newblock In {\em Measures of complexity: festschrift for alexey chervonenkis}, pages 11--30. Springer, 2015.

\bibitem{weimann2015approximating}
Oren Weimann and Raphael Yuster.
\newblock Approximating the diameter of planar graphs in near linear time.
\newblock {\em ACM Transactions on Algorithms (TALG)}, 12(1):1--13, 2015.
\newblock \href {https://doi.org/10.1145/2764910} {\path{doi:10.1145/2764910}}.

\end{thebibliography}
\appendix

\section{Missing proofs}
\label{app:missing}

\lemLocal*
\begin{proof}
    We fix an arbitrary vertex $r \in V(G)$ and compute the shortest-path tree $T$ rooted at $r$ using Dijkstra's algorithm in time $\Oh(n\log n)$.
    For a vertex $v \in V(G)$ we define $\layer(v)$ as the largest integer $j$ for which $j\cdot \delta \le d_T(r,v)$.
    We construct the clustering $\ccal$ by contracting each edge in $T$ connecting two vertices of the same layer.
    The layer of such cluster is defined as the common layer of the vertices inside it.
    In each cluster $C_i$ there is a unique vertex $v$ whose layer is greater than the layer of its parent in $T$ or $v=r$: this vertex becomes the center $c_i$ of the cluster $C_i$.
    Clearly $c_i$ can be reached from each $v \in C_i$ using a path contained in $T$ of length at most $\delta$.
    This path goes inside $G[C_i]$ so indeed we obtain a $\delta$-radius clustering.

    (\ref{item:local:interval}). Let $I = [i,j]$ and $p = j-i+1$.
        Observe that for each cluster $C$ with $\layer(C) < i$ there is a path in $G_\ccal$ from the root cluster to $C$ that goes only through clusters of layers smaller than $i$.
        Hence the subgraph of  $G_\ccal$ induced by such layers is connected.
        If $i > 0$ let ${\hat G}^I_\ccal$ be obtained from $G_\ccal$ by contracting into a single vertex all clusters $x$ with $\layer(x) < i$ and removing all clusters $y$ with $\layer(y) > j.$
        If $i = 0$ we set ${\hat G}^I_\ccal = { G}^I_\ccal$.
        
        We claim that the diameter of the unweighted graph  ${\hat G}^I_\ccal$ is at most $2p$.
        Consider a cluster $x$ for which $\layer(x) > i$.
        Then $x$ is adjacent in $G_\ccal$ to a cluster $y$ for which $\layer(y) = \layer(x) - 1$.
        Hence for each $x \in V({\hat G}^I_\ccal)$ there is path, on at most $p$ vertices, to $y$ with $\layer(y) = i$.
        If $i = 0$ then there is unique cluster of layer 0 and so each pair of vertices in ${\hat G}^I_\ccal$ can be connected by a path of length at most $2p-1$.
        If $i > 0$ then every cluster is within distance $p$ from the super node obtained from contracting layers smaller then $i$ and so the (unweighted) diameter of ${\hat G}^I_\ccal$ is bounded by $2p$.
        
        Since ${\hat G}^I_\ccal$ is a contraction of $G$ it is also $H$-minor-free.
        Due to \Cref{lem:local-property}, the treewidth of ${\hat G}^I_\ccal$ can be bounded by $\Oh_H(p)$.
        The claim follows by observing that ${ G}^I_\ccal$ is a subgraph of ${\hat G}^I_\ccal$ so its treewidth cannot be larger.
    
    (\ref{item:local:tw}). 
     We prove the bound on $\tw(G_\ccal)$.
     Let $k$ be the highest assigned layer number.
     For each $v \in V(G)$ we have $d_T(r,v) = d_G(r,v) \le \diam(G)$ so $k \le \diam(G) / \delta$.
     The bound follows from point~(\ref{item:local:interval}) for interval $I = [0,k]$.

    (\ref{item:local:distance}).
     Consider $u,v \in V(G)$ with  $d_G(u,v) \le \delta \cdot p$.
     From the triangle inequality we get
     \[|d_T(r,u) - d_T(r,u)| = |d_G(r,u) - d_G(r,v)| \le d_G(u,v). \]
     
     Hence $|\layer(u) - \layer(v)| \le p$.
    Assume w.l.o.g. that $\layer(u) \le \layer(v)$ and let $I = [\layer(v) - p, \layer(u) + p]$; we have $|I| \le 2p + 1$.
    Every vertex $w$ lying on a shortest $(u,v)$-path in $G$ is also within distance $\delta \cdot p$ from each of $u,v$ so it must satisfy $\layer(w) \in I$.
    Therefore, any shortest $(u,v)$-path in $G$ is also present in $G^I$.
\end{proof}

\subsection{Eccentricities}
\label{app:ecce}

This subsection contains a self-contained proof of \Cref{item:thm:ecce} from \Cref{thm:preprocessing}.

For a non-root $t \in V(T)$ we refer to its parent in $T$ as $\parent(t)$.
For $u,v \in V(G)$ let $\apxdist(u,v)$ denote the value returned by the oracle from \Cref{claim:oracle} when queried for the pair $u,v$.
This is a well-defined function because the oracle is deterministic.
Recall the concept of an additive core-set from \Cref{lem:corest-exist}.
Also, recall the sets $V_t$ from the first part of the proof and that $\tau = \diam(G) / \delta$.

\claimCoreset*
\begin{claimproof}
     Fix a non-root $t \in V(T)$.   
     Let us index the set $\chi(\parent(t))$ in any order to form a tuple $S = (s_0, \dots, s_{k-1})$.
     Note that $k \le w + 1$.
     We call the oracle from \Cref{claim:oracle} for each $v \in V_t$ and each $s_j$, $0 \le j < k$.
     A single call requires time  $\Oh(w\cdot \log n)$ so we spend
     $\Oh( w^2\cdot \log n\cdot |V_t|)$ time. 
     
     For each $v$ we now have a vector ${\bf y}_v \in \rr^k$ defined as ${\bf y}_v[j] = \apxdist(v,s_j)$.
     Recall that $\Phi_{S}(v)$ stands for the distance tuple: $\br{d_G(v,s_0), \dots, d_G(v,s_{k-1})}$.
     The oracle specification guarantees that $\ell_\infty({\bf y}_v - \Phi_{S}(v)) \le 2\delta$.
     
     We apply 
    \Cref{lem:metric:apx-cover} for $\delta' = 2\delta$ to compute a $(10\cdot\delta)$-additive core-set $V'_t \sub V_t$ of size $\Oh\br{w^{h-1}\tau^{h}}$ in time $\Oh\br{(w\cdot\tau)^{h} \cdot |V_t|}$.
    The bound on the total running time follows from \Cref{claim:sum-vt}.
\end{claimproof}

\begin{claim}\label{claim:separator}
        Consider $t \in V(T)$ with distinct children $t_1, t_2$.
        Let $u_1 \in V_{t_1}, u_2 \in  V_{t_2}$.
        Then there is $i \in \chi(t)$ for which $d_{G}(u_1,c_i) + d_{G}(c_i,u_2) \le d_G(u_1,u_2) +2\delta$. 
    \end{claim}
    \begin{claimproof}
        Let $j_1,j_2 \in \ccal$ denote the clusters of $u_1,u_2$, respectively.
        We have $j_1 \in U_{t_1}, j_2 \in U_{t_2}$ and this in particular implies $j_1,j_2 \not\in\chi(t)$.
        Since all the bags containing $j_1$ (resp. $j_2$) are located within the subtree rooted at $t_1$ (resp. $t_2$), we know that $\chi(t)$ is a $(j_1,j_2)$-separator in $G_\ccal$.
        As $G_\ccal$ is a contraction of $G$, it follows that $\bigcup_{i \in \chi(t)} C_i$ is a $(u_1,u_2)$-separator in $G$.

        Let $Q$ be a shortest $(u_1,u_2)$-path in $G$.
        By the argument above, there is some $i \in \chi(t)$ for which $Q$ intersects $C_i$.
        The claim follows from \Cref{lem:path-touch}.
    \end{claimproof}

In the next claim, we argue that the core-sets suffice to estimate the maximal distance between $v$ and $u$, where $v$ is fixed and $u$ is quantified over all vertices from some subtree. 

\claimAnal*
 \begin{claimproof}
        We will first show that $y \ge x - 6\delta$.
        Let $u^* \in V'$ be the vertex that realizes $x$.
        By the specification of the oracle (\Cref{claim:oracle}), the value of $\apxdist(v,c_i)$ is upper bounded by $d_G(v,c_i) + 2\delta$ for each $i \in \chi(t)$ and likewise for $u^*$. 
        Hence $x \le \min_{i \in \chi(t)} d_G(v,c_i) +  d_G(c_i,u^*) + 4\delta$.    
        By \Cref{claim:separator}, there is $i \in \chi(t)$ such that 
        \[d_G(v,u^*) \ge d_G(v,c_i) + d_G(c_i,u^*) - 2\delta \ge x - 6\delta.\]
        As $u^*$ belongs to $V' \sub V_{t_2}$, its distance to $v$ gives a lower bound for $y$.
        
        Now we show $y \le x + \rho$.
        Let $u \in V_{t_2}$ be the vertex farthest from $v$. 
        By the definition of an additive core-set, there is $u' \in V'$ for which $\ell_\infty(\Phi_{\chi(t)}(u) - \Phi_{\chi(t)}(u')) \le \rho$.
        Let $i \in \chi(t)$ minimize $\apxdist(v,c_i) +  \apxdist(c_i,u')$.
        The oracle never underestimates distances (\Cref{claim:oracle})~so
        \[d_G(v,c_i) + d_G(c_i,u') \le \apxdist(v,c_i) +  \apxdist(c_i,u') \le x. \]
        Next, $d_G(c_i,u) \le d_G(c_i,u') + \rho$, which gives
        \[y = d_G(v,u) \le d_G(v,c_i) + d_G(c_i,u) \le d_G(u,c_i) + d_G(c_i,u') + \rho \le x + \rho. \]
        In the proof of the first inequality, we showed that $d_G(v,u^*)  \ge x - 6\delta$ where $u^* \in V'$ is the vertex realizing $x$.
        Together with the second inequality  $x \ge y - \rho$, this yields the second part of the claim.
    \end{claimproof}

We will now use the additive core-sets precomputed in \Cref{claim:corset-compute} to quickly find eccentricity of a vertex.

\claimEcce*
\begin{claimproof}
        We consider three cases depending on the location of the cluster of the farthest vertex~$u$.
        Then we will output maximum over these three options. 
        Let $i_v,i_u \in \ccal$ denote the clusters containing $v,u$, respectively, and $c_v,c_u$ denote their centers.
        Also, let $t = \loc(v)$. 
        \begin{enumerate}

            \item The node $\loc(u)$ lies on the $t$-to-root path in $T$, inclusively 
            We inspect each $t' \in V(T)$ on this path and each $i \in \chi(t')$, then we output $c_i$ that maximizes $\apxdist(v,c_i)$.
            Since the oracle never underestimates distances, $\apxdist(v,c_i) \ge d_G(v,c_i)$.
            By the triangle inequality, $d_G(v,c_u) \ge d_G(v,u) - \delta$ so the found center $c_i$ satisfies $\apxdist(v,c_i) \ge \ecc_G(v) - \delta$. 
            Next, by \Cref{claim:oracle} the true $(v,c_u)$-distance is no smaller than  $\apxdist(v,c_u) - 2\delta \ge \ecc_G(v) - 3\delta$.
            We perform $\Oh(w\cdot \log n)$ oracle calls, each of cost $\Oh(w\cdot \log n)$.

            \item The node $\loc(u)$ is a descendant of $t$. Then $\loc(u) \in U_{t'}$ where $t'$ is a child of $t$.
            Recall than $T$ is a binary tree so there are at most two options for $t'$; from now on consider $t'$ fixed.
            Note that $u \in V_{t'}$ and $i_v \in \chi(t) = \chi(\parent(t'))$ so we can take advantage of the precomputed $(10\cdot\delta)$-additive core-set $V'_{t'} \sub V_{t'}$ with respect to $\chi(t)$ 
            (\Cref{claim:corset-compute}).
            We output $u' \in V'_{t'}$ that maximizes $\apxdist(c_v,u')$.
            We get that 
            \[\apxdist(c_v,u') \ge d_G(c_v,u') \ge d_G(c_v,u) - 10\cdot\delta \ge d_G(v,u) - 11\cdot\delta, \]
            where the last estimation follows from the triangle inequality.
            Again, the true $(c_v, u')$-distance is no smaller than  $\apxdist(c_v,u') - 2\delta \ge \ecc_G(v) - 13\cdot \delta$.
            The number of necessary oracle calls is proportional to the size of the core-set, which is $\Oh(w^{h-1}\tau^h)$.
            The total running time is then $\Oh\br{(w\cdot\tau)^{h}  \log n}$.

            \item It remains to consider $\loc(u)$ that is neither ancestor nor descendant of $t$.
            Then there exists $\widehat{t} \in V(T)$ so that $u, v$ reside in different subtrees of $\widehat{t}$.
            More precisely, $i_v$ lies in the subtree rooted at $t_1$ and $i_u$ lies in the subtree rooted at $t_2$, where $t_1,t_2$ are distinct children of $\widehat{t}$.
            Consequently $v \in V_{t_1}$ and $u \in V_{t_2}$.
            There are $\Oh(\log n)$ choices for $\widehat{t}$; suppose from now on that it is fixed.

            We use the precomputed $(10\cdot \delta)$-additive core-set $V'_{t_2} \sub V_{t_2}$ with respect to $\chi(\widehat{t})$(\Cref{claim:corset-compute}).
            By applying \Cref{claim:corset-analysis} for $\rho = 10\cdot \delta$,
            there is $u' \in V'_{t_2}$ for which the value of 
            \[x = \min_{i \in \chi(\widehat t)} (\apxdist(v,c_i) +  \apxdist(c_i,u'))\]
            satisfies $|x - \ecc_G(v)| = |x - d_G(v,u)| \le 16\cdot \delta$. 
            Moreover, $u'$ is the vertex from $V'_{t_2}$ that maximizes the above formula.
            In order to find $u'$, we need to
            inspect each $u'' \in V'_{t_2}$ and ask the oracle for all $(u'',c_i)$-distances, $i \in \chi(\widehat t)$.
           This takes time $\Oh(w^2\log n \cdot |V'_{t_2}|) = \Oh\br{w^{h+1}\cdot\tau^{h}  \log n}$.      
            As shown in \Cref{claim:corset-analysis}, the true $(v, u')$-distance is no less than $\ecc_G(v) - 16\cdot\delta$.

            There are $\Oh(\log n)$ choices for $\widehat{t}$ and so the final running time gets multiplied by $\Oh(\log n)$.  
        \end{enumerate}
    \end{claimproof}

\Cref{thm:preprocessing}(\ref{item:thm:ecce}) follows by applying \Cref{claim:ecce} to each vertex.

\section{Proofs of Theorems 3 and 4}
\label{app:main}

In order to design a distance oracle with multiplicative error, we need one auxiliary lemma.

\begin{lemma}\label{lem:oracle-delta}
    Let $H$ be an apex graph.
    Given an edge-weighted $n$-vertex $H$-minor-free graph $G$ and  $\eps, \Delta \in \rr_+$, we can in time $\Oh_H(\eps^{-8}\cdot n\log n)$ 
    build a data structure that supports the following queries.
    For given $u,v \in V(G)$ it either correctly concludes that $d_G(u,v) > \Delta$ or
     outputs a number 
   $x \in [d_G(u,v),\, d_G(u,v) + \eps\Delta/2]$.
    A single query is processed in time $\Oh_H(\eps^{-2}\cdot\log n)$.
\end{lemma}
\begin{proof}
    We apply \Cref{lem:local-tw} for $\delta = \eps\Delta/4$ in time $\Oh(n \log n)$.
    Let $\ccal$ denote the obtained $\delta$-radius clustering of $G$ with layering function $\layer$, and $k$ denote the highest assigned layer number.
    Set $p = \lceil 4/\eps \rceil$. 
    Let ${\cal I}$ be the family of all intervals of length $2p+1$ 
    that intersect the interval $[0,k]$.
    For each $I \in {\cal I}$ we know from \Cref{lem:local-tw}(\ref{item:local:interval}) that $\tw(G^I_\ccal) = \Oh_H(p)$.
    We use \Cref{lem:prelim:twd-compute} to build a balanced tree decomposition of $\tw(G^I_\ccal)$ of width $\Oh_H(p^2)$; this takes time $\Oh_H(|V(G^I)| \cdot p^7\log n)$.
    Such a decomposition
    can be supplied to 
    \Cref{thm:preprocessing}(\ref{item:thm:oracle}):
    in time $\Oh_H(|V(G^I)| \cdot p^2\log n)$ we build a distance oracle that given $u,v \in V(G^I)$ returns their distance in $G^I$ with additive positive error $2\delta = \eps\Delta/2$ in time $\Oh(p^2\log n)$.

    We analyze the running time of setting up all the oracles.
    For each vertex $v \in V(G)$ there are at most $2p+1$ intervals $I \in {\cal I}$ for which $\layer(v) \in I$.
    Hence $\sum_{I \in {\cal I}} |V(G^I)| \le (2p+1)\cdot n$.
    Therefore, the total time needed to build all the oracles is $\Oh_H(p^8\cdot  n\log n)$.

    Suppose now that we are asked for a distance between $u,v \in V(G)$.
    Due to \Cref{lem:local-tw}(\ref{item:local:distance}), if $d_G(u,v) \le \Delta \le p \cdot \delta$ then $|\layer(u) - \layer(v)| \le p$.
    If the latter condition does not hold, we can answer that  $d_G(u,v) > \Delta$.
    Otherwise, \Cref{lem:local-tw}(\ref{item:local:distance}) states that there is $I \in {\cal I}$, depending only on $(\layer(u),\layer(v),p)$ for which the $(u,v)$-distance is the same in $G$ and $G^I$.
    We can therefore use the prepared oracle for $G^I$ to compute $x \in [d_G(u,v),\, d_G(u,v) + \eps\Delta/2]$ in time  $\Oh(p^2\log n)$.
\end{proof}

\thmOracle*
\begin{proof}
    We can assume that $W$ is approximately known due to 2-approximation of diameter.
    Next, we can assume that the smallest distance equals 1 by scaling.
    Let $\ell = \lceil \log_2 W \rceil$.
    For each $i \in [\ell]$ we build the data structure from \Cref{lem:oracle-delta} for $\Delta = 2^{i}$.

    Consider a query $(u,v)$.
    We ask all the $\ell$ oracles for the $(u,v)$-distance and output the smallest answer.
    We claim that this is a good approximation for $d_G(u,v)$.
    Let $i \in [\ell]$ be such that $d_G(u,v) \in [2^{i-1}, 2^i]$.
    Then the oracle for $\Delta = 2^{i}$  returns $x \in [d_G(u,v), \, d_G(u,v) + \eps\cdot 2^{i-1}]$.
    This is a $(1+\eps)$-approximation because $2^{i-1} \le d_G(u,v)$.
    Finally, observe that no oracle can underestimate the $(u,v)$-distance so the smallest answer is guaranteed to lie inside this interval.
\end{proof}

We move on to the last theorem.

\thmApexDiam*
\begin{proof}
    Let $H$ of size $h$ be the apex graph excluded by $\cal H$.
    Because $G$ must be $K_{h+c}$-minor-free, we can use the linear  algorithm for computing single-source distances~\cite{TazariM09}.
    First we compute the eccentricity of an arbitrary vertex $v$ in linear time and set $D_0 = \ecc_G(v)$.
    We know that $D_0 \le \diam(G) \le 2D_0$.
    Let $k = \lceil 1/\eps \rceil$ and $D_0 < D_i < \dots < D_k = 2D_0$ divide the interval $[D_0,\, 2D_0]$ evenly. 

    \begin{claim}\label{claim:apx-step}
        There is an algorithm that, given the input to \Cref{thm:apex}  and $D \in [D_0,\, 2D_0]$, runs in time $\Oh_{h+c}(\eps^{-\alpha}\cdot n\log^2 n)$, where $\alpha = \Oh((h+1)(c+1))$, and correctly reports that either $\diam(G) > D$ or $\diam(G) \le  (1 + \eps)\cdot D$.
    \end{claim}
    \begin{claimproof}
        Note that $2D \ge \diam(G)$ so we are allowed to
        apply \Cref{lem:diam-reduction} with $\eps' = \eps/3$.
        If we receive the message that $\diam(G) > D$, we propagate it.
        Otherwise, let $(G', A')$ be the output of the algorithm.
        We have $\diam(G') \ge \diam (G)$ and, assuming $\diam(G) \le D$, also $\diam(G') \le (1 + \frac\eps 3)\cdot D$.
        Let $\delta = \eps D / 48$.
        Each connected component $G^i$ of $G'- A'$ 
        satisfies $\diam(G^i) = \Oh_c(\eps^{-c}\cdot D)$ so  $\diam(G^i) / \delta = \Oh_c(\eps^{-(c+1)})$.
        Next, $G^i$
        is $H$-minor-free and so we can apply \Cref{lem:local-tw} to get a $\delta$-radius clustering $\ccal^i$ of $G^i$ so that the treewidth of the cluster graph $G^i_{\ccal^i}$ is $\Oh_h(\diam(G^i) / \delta) = \Oh_{h+c}(\eps^{-(c+1)})$.
        We combine those to get a $\delta$-radius clustering $\ccal'$ of $G'$ by assigning each $w \in A'$ a singleton cluster.
        Observe that the treewidth of the combined cluster graph grows by at most $|A'| = c$
        with respect to $\max(\tw(G^i_{\ccal^i}))$
        because one can simply insert the additional $c$ cluster-vertices to all the bags. 
        
        We use \Cref{lem:prelim:twd-compute} to compute a tree decomposition $(T,\chi)$ of $G'_{\ccal'}$ of width $w = \Oh_{h+c}(\eps^{-2(c+1)})$, depth $\Oh(\log n)$, and size $\Oh(n)$.
        This takes time $\Oh_{h+c}(\eps^{-7(c+1)}\cdot n\log n)$.
        Note that $\tau := \diam(G')/\delta$ is of order $\Oh(1/\eps)$.
        The  decomposition $(T,\chi)$ can be processed with \Cref{thm:preprocessing}(\ref{item:thm:ecce}): in time $\Oh_{h+c}(w^{h+1}\tau^h\cdot n\log^2 n) =  \Oh_{h+c}(\eps^{-2(c+2)(h+1)}\cdot n\log^2 n)$ we can estimate $\diam(G')$ within additive error $16\cdot \delta = \eps D / 3$.
        
        Let $\Delta$ denote the returned value.
        If $\Delta > (1 + \frac{2\eps}{3})\cdot D$ then $\diam(G') >(1+\frac{\eps}{3})\cdot D$ so by the guarantee of \Cref{lem:diam-reduction} we know that $\diam(G) > D$ and we can terminate the algorithm with such a message.
        Otherwise $\Delta \le (1 + \frac{2\eps}{3})\cdot D$ and the bound $\diam(G) \le (1 + \eps) \cdot D$ follows from $\diam(G) \le \diam(G') \le \Delta + \eps D / 3$.
    \end{claimproof}

    We run the algorithm from \Cref{claim:apx-step} for $D = D_0,D_1,\dots,D_k$ until we receive information that $\diam(G) \le (1+\eps)\cdot D_i$. 
    Observe that we must witness such $i$ at some point because $\diam(G) \le D_k$ so the algorithm cannot always terminate with the first message type. 
    We claim that $D_i$ approximates $\diam(G)$ within factor $(1+\eps)$. 
    
    First suppose $i = 0$.
    By the choice of $D_0$ we know that $\diam(G) \ge D_0$. 
    Then the halting condition $\diam(G) \le (1+\eps)\cdot D_0$ implies that $D_0$ is a good estimation of the diameter.

     Now suppose that $i > 0$.
     Because we have not terminated the algorithm in the previous step,
     we must have learned that $\diam(G) > D_{i-1} \ge (1 - \eps) \cdot D_i$.
     The halting condition $\diam(G) \le (1+\eps)\cdot D_i$ gives the second inequality.

    The number of times we apply \Cref{claim:apx-step} is $k = \lceil 1/\eps \rceil$, which increases the exponent by one.
    The theorem follows.
\end{proof}

\end{document}